\documentclass[journal]{IEEEtran}
\usepackage{graphicx}
\usepackage{cite}
\usepackage{amsmath,amsthm,amssymb}
\usepackage{subfigure}

\newtheorem{thm}{Theorem}[section]

\newtheorem{definition}{Definition}
\newtheorem{remark}{Remark}

\IEEEoverridecommandlockouts

\ifCLASSINFOpdf
\else
\fi
\begin{document}
\title{Relay-assisted Multiple Access with Full-duplex Multi-Packet Reception}

\author{\IEEEauthorblockN{Nikolaos Pappas~\IEEEmembership{Member,~IEEE}, Marios Kountouris~\IEEEmembership{Senior~Member,~IEEE}, Anthony Ephremides~\IEEEmembership{Life~Fellow,~IEEE}, Apostolos Traganitis\thanks{Manuscript received September 19, 2014; revised January 30, 2015 and accepted February 19, 2015.}
\thanks{N. Pappas is with the Department of Science and Technology, Link\"{o}ping University, Norrk\"{o}ping SE-60174, Sweden (e-mail: nikolaos.pappas@liu.se).
M. Kountouris is with the Mathematical and Algorithmic Sciences Lab, France Research Center, Huawei Technologies Co. Ltd. (e-mail: marios.kountouris@huawei.com).
A. Ephremides is with the Department of Electrical and Computer Engineering and Institute for Systems Research University of Maryland, College Park, MD 20742 (e-mail: etony@umd.edu).
A. Traganitis is with the Computer Science Department, University of Crete, Greece and Institute of Computer Science, Foundation for Research and Technology - Hellas (FORTH) (e-mail: tragani@ics.forth.gr).
}
\thanks{This work has been partially supported by the People Programme (Marie Curie Actions) of the European Union's Seventh Framework Programme FP7/2007-2013/ under REA grant agreement no.[612361] -- SOrBet and by the NSF grant CCF1420651, ONR grant N000141410107.}
\thanks{This work was presented in part in IEEE Information Theory Workshop 2011\cite{b:Pappas-ITW2}.}}
}

\markboth{IEEE Transactions on Wireless Communications,~Vol.~XX, No.~X, MONTH~YEAR}{Pappas \MakeLowercase{\textit{et al.}}: Relay-assisted Multiple Access with Full-duplex Multi-Packet Reception}

\maketitle

\begin{abstract}
The effect of full-duplex cooperative relaying in a random access multiuser network is investigated here. First, we model the self-interference incurred due to full-duplex operation, assuming multi-packet reception capabilities for both the relay and the destination node. Traffic at the source nodes is considered saturated and the cooperative relay, which does not have packets of its own, stores a source packet that it receives successfully in its queue when the transmission to the destination has failed. We obtain analytical expressions for key performance metrics at the relay, such as arrival and service rates, stability conditions, and average queue length, as functions of the transmission probabilities, the self interference coefficient, and the links' outage probabilities. Furthermore, we study the impact of the relay node and the self-interference coefficient on the per-user and aggregate throughput, and the average delay per packet. We show that perfect self-interference cancelation plays a crucial role when the SINR threshold is small, since it may result to worse performance in throughput and delay comparing with the half-duplex case. This is because perfect self-interference cancelation can cause an unstable queue at the relay under some conditions.
\end{abstract}


\section{Introduction}
\label{sec:intro}
Driven by the exponential traffic growth and the ever-increasing demands for wider spectrum, the quest for higher spectral efficiency and enhanced reliability and coverage is creating a new impetus for cooperative communication systems. Cooperative communication aims at increasing the link data rates and the reliability of time-varying links, by overcoming fading and interference in wireless networks. Among the various cooperation techniques to increase throughput, full-duplex relaying has recently gained significant attention. The vast majority of research papers have considered half-duplex or out-of-band full-duplex systems, in which terminals cannot transmit and receive at the same time, or over the same frequency band. However, the use of nodes with in-band full-duplex capability, i.e. terminals that transmit and receive simultaneously over the same frequency band, is constantly increasing in current wireless networks as they can potentially double the network spectral efficiency. Moreover, full-duplex relay systems open a whole new spectrum of capabilities, such as collision detection in contention-based networks. In this work, we focus on a relay-assisted random access network and we analyze the effect of full-duplex cooperative relaying in the network performance, namely arrival and service rates, stability conditions, and average queue length at the relay.

\subsection{Related Work}

The classical relay channel was originally introduced by van der Meulen~\cite{b:Muelen}, and earlier work on the relay channel was based on information-theoretic formulations, e.g. ~\cite{b:CoverGamal}. Most cooperative techniques that have been studied so far focus on the benefits of physical layer cooperation~\cite{b:Yates-NOW}. Nevertheless, there is evidence that the same gains can be achieved with network layer cooperation, which is plain relaying without any physical layer considerations~\cite{b:Sadek, b:Rong1}. Recently several works have investigated relaying performance at the MAC layer~\cite{b:Sadek, b:Rong1, b:Rong2, b:Simeone, b:Pappas, b:PappasITW2013Relay, b:PappasGlobalsip2013, b:PappasTIT2015}. More specifically, in~\cite{b:Sadek}, the authors have studied the impact of cooperative communication at the medium access control layer with TDMA. They introduced a new cognitive multiple access protocol in the presence of a relay in the network. In~\cite{b:Pappas-ISIT} the notion of partial network level cooperation is introduced by adding a flow controller at the relay, which regulates the amount of provided cooperation depending on the conditions of the network. The classical analysis of random multiple access schemes, like slotted ALOHA~\cite{b:Bertsekas}, has focused on the so-called collision model. Random access with multi-packet reception (MPR) has attracted attention recently~\cite{b:AlohaVerdu, b:Angel, b:Naware, b:PappasMPR}. All the above approaches come together in the model that we consider.

In wireless networks, when a wireless node transmits and receives simultaneously in the same frequency, the problem of self-interference arises. Self-interference mitigation is a key challenge in in-band full-duplex systems. Information-theoretic aspects of this problem can be found in the pioneering work of Shannon~\cite{b:Shannon2way}, although the capacity region of the two-way channel is not known for the general case~\cite{b:Cover}. The information-theoretic limits of in-band full-duplex relaying have been studied focusing on the idealistic case of perfect self-interference cancelation \cite{Somekh07, Jafar09}. There exist several techniques that allow the possibility of perfect self-interference cancelation~\cite{b:Cover}. However, in practice, there are several technological limitations and challenges~\cite{b:selflimit1, b:selflimit2}, which may limit the accuracy and the effectiveness of self-interference cancelation. Various methods for performing self-interference cancelation at the receivers can be found in~\cite{b:selfcancel1} and~\cite{b:selfcancel2}. The main result therein is that there is a tradeoff between transceiver complexity and self-interference cancelation accuracy. In ~\cite{b:Tsubouchi93, b:Chen98}, it was demonstrated in practice real implementations of simultaneous transceivers,  where the self-interference problem has been mitigated through RF isolators and echo cancellers, coupled with base-band digital filtering. Furthermore, some recent results have also shown that full duplex is possible, proposing specific designs, e.g.~\cite{b:Choi-Mobicom2010, b:Jain-Mobicom2011}, which mainly focus on the physical and the medium access control (MAC) layer design. Choi et al. in~\cite{b:Choi-Mobicom2010} designed a practical single-channel full-duplex wireless system, combining three self-interference cancellation schemes, as well as RF and digital interference cancellation. Jain et al.~\cite{b:Jain-Mobicom2011} presented a full-duplex radio design using signal inversion and adaptive cancellation. Unlike~\cite{b:Choi-Mobicom2010}, the authors in \cite{b:Jain-Mobicom2011} consider wideband and high power systems. In theory, this new design has no limitation in terms of bandwidth or power. Therefore, building full-duplex wireless networks (such as full-duplex 802.11n wireless networks) has started becoming feasible. Fang et al.~\cite{b:Pathbook} proposed a collision-free full-duplex broadcast MAC and studied cross-layer optimization of MAC and routing in full-duplex wireless networks under various resource and social constraints. In~\cite{b:KwonTVT12} the comparison of performance of half and full-duplex relay is studied at the physical layer, in~\cite{b:KimTVT13} is investigated the effect of channel estimation errors on the ergodic capacities for bidirectional full-duplex transmission. An information theoretic study in~\cite{b:ErkipCISS13} compares multi-antenna half and full-duplex relaying from the perspective of achievable rates.

\subsection{Contribution}

In this work, we complement and extend the work in \cite{b:Pappas-ITW2}. We study the operation of a cooperative node relaying packets from a number of users/sources to a destination node as shown in Fig.~\ref{fig:netmodel}. We assume MPR capability for both the relay and the destination node. The relay node can transmit and receive at the same time over the same frequency band (in-band full duplex). We assume random medium access, slotted time, and that each packet transmission takes one timeslot. The wireless channel is modeled as Rayleigh flat-fading channel with additive white Gaussian noise. A user transmission is successful if the received signal-to-interference-plus-noise ratio (SINR) is above a certain threshold $\gamma$. We also assume that acknowledgements (ACKs) are instantaneous and error free. The relay does not have packets by itself and the source nodes are considered saturated with unlimited amount of traffic. The self-interference cancellation at the relay is modeled as a variable power gain, mainly because we are studying the impact on the network layer\footnote{The self-interference cancellation at the relay is modeled as a variable power gain that affects the success probability with which the relay will receive a packet and is described in Section~\ref{sec:sysmod}.}. Studying in detail the physical layer implementation of self-interference mitigation and considering specific self-interference cancelation mechanisms is beyond the scope of this paper.
We obtain analytical expressions for key performance characteristics of the relay queue, such as arrival and service rates, and we derive conditions for stability and the average queue length as functions of the transmission probabilities, the self-interference coefficient, and the links' outage probabilities. In particular, we study the impact of the relay node and the self-interference coefficient on the per-user and the network-wide throughput, as well as the average delay per packet. Furthermore, we derive expressions for both the per-user and aggregate throughput when the queue at the relay is unstable, for which case we do not have though any guarantees for bounded delay.

The remainder of the paper is organized as follows: Section~\ref{sec:sysmod} describes the system model and in Section~\ref{sec:analysis} we present the main characteristics of the relay queue, such as the average arrival and service rates. In Section~\ref{sec:thr_analysis}, we provide expressions for the per-user and the aggregate throughput. The average delay per packet is obtained in Section~\ref{sec:delay_analysis}. Numerical results are presented in Section~\ref{sec:results}, and finally Section~\ref{sec:conclusions} concludes the paper.

\section{System Model}
\label{sec:sysmod}
\subsection{Network Model}
We consider a network with $n$ sources, one relay node, and a single destination node. The sources transmit packets to the destination using a cooperative relay; the case of $n=2$ is depicted in Fig.~\ref{fig:netmodel}. We assume that the queues of both sources are saturated, i.e., no external arrivals and unlimited buffer size, and that the relay does not have packets of its own but only forwards the packets it has received from the two users. The relay node stores a source packet that it receives successfully in its queue when the direct transmission to the destination node has failed. We assume a random access channel where $q_{0}$ is the transmit probability of the relay given that it has packets in its queue, and $q_{i}$ for $i \neq 0$ is the transmit probability for the $i$-th user. The receivers at the relay and the destination nodes are equipped with multiuser detectors, hence they can decode packets from more than one transmitter at a time. Furthermore, the relay can simultaneously transmit and receive packets (full duplex).

\begin{figure}[]
\centering
\includegraphics[scale=0.6]{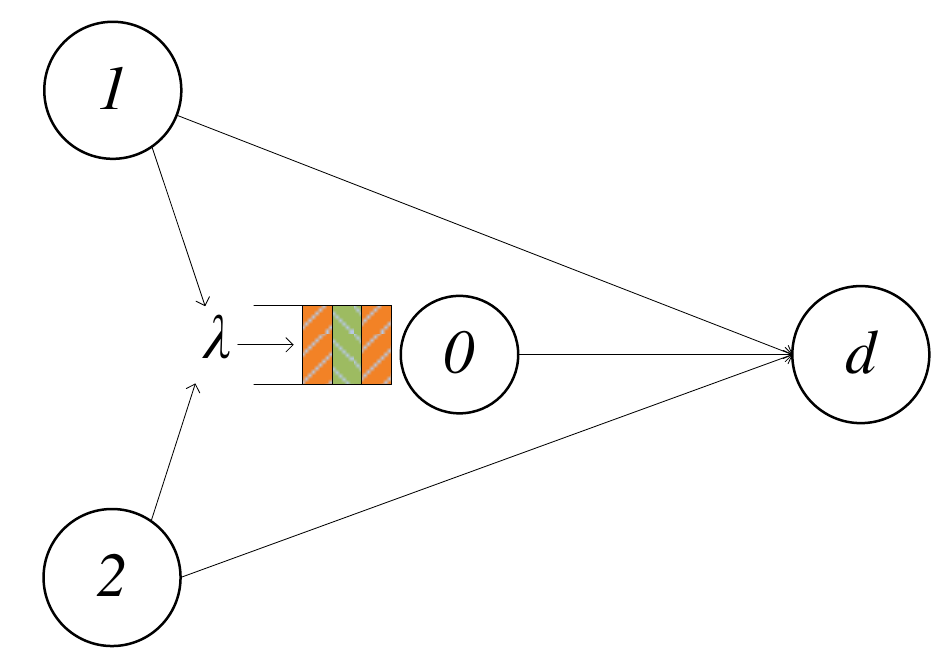}
\caption{The network model for the two-user case: users have saturated queues and the relay only forwards the packets received from both users, which failed to reach the destination.}
\centering
\label{fig:netmodel}
\end{figure}

\subsection{Physical Layer Model}
The MPR channel model used in this paper is a generalized form of the packet erasure model. We assume that a packet transmitted by node $i$ is successfully received by node $j$ if and only if ${\rm SINR}(i,j)\geq \gamma_{j}$, where $\gamma_{j}$ is a threshold characteristic of node $j$. The wireless channel is subject to fading; let $P_{tx}(i)$ be the transmit power at node $i$ and $r(i,j)$ be the distance between $i$ and $j$. The received power at $j$ when $i$ transmits is $P_{rx}(i,j)=A(i,j)h(i,j)$ where $A(i,j)$ is a random variable representing small-scale fading. Under Rayleigh fading, $A(i,j)$ is exponentially distributed~\cite{b:Tse}. The received power factor $h(i,j)$ is given by $h(i,j) = P_{tx}(i)(r(i,j))^{-\alpha}$ where $\alpha$ is the path loss exponent with typical values between $2$ and $6$. We model the self-interference by a scalar $g \in [0,1]$ as in~\cite{b:CodreanuWiOpt10} and~\cite{b:CodreanuITW10}. We refer to the $g$ as the self-interference coefficient. When $g = 1$, no self-interference cancelation technique is used, while $g = 0$ models perfect self-interference cancelation.
The success probability in the link $ij$ is given by
\begin{equation}
\begin{aligned}
\label{eq:succprob}
P_{i/\mathcal{T}}^{j}=\exp\left(-\frac{\gamma_{j}\eta_{j}}{v(i,j)h(i,j)}\right) \left(1+\gamma_{j}(r(i,j))^{\alpha}g \right)^{-m} \times \\ \times \prod_{k\in \mathcal{T}\backslash \left\{i,j\right\}}{\left(1+\gamma_{j}\frac{v(k,j)h(k,j)}{v(i,j)h(i,j)}\right)}^{-1},
\end{aligned}
\end{equation}
where $\mathcal{T}$ is the set of transmitting nodes at the same time, $v(i,j)$ is the parameter of the Rayleigh fading random variable, $\eta_{j}$ is the receiver noise power at $j$ and $m=1$ when $j \in \mathcal{T}$ and $m=0$ else. The analytical derivation for this success probability can be found in~\cite{b:Tse}.

\emph{Note:} The self-interference is modeled through $g$ and it affects the success probability when the relay transmits and receives simultaneously.
The value of $g$ captures the accuracy of the self-interference cancelation. As $g$ approaches $0$ it is closer to the pure full duplex operation. 
When $g$ is $1$ the operation is the half duplex operation since the success probabilities for the users in this case are very close to $0$.

\subsection{Queue Stability}
We adopt the definition of queue stability used in~\cite{Szpankowski:stability}.

\begin{definition}
Denote by $Q_i^t$ the length of queue $i$ at the beginning of timeslot $t$. The queue is said to be \emph{stable} if
\begin{equation}\label{eqn:definition_stability}
    \lim_{t \rightarrow \infty} {Pr}[Q_i^t < {x}] = F(x)  \text{ and } \lim_{ {x} \rightarrow \infty} F(x) = 1.
\end{equation}

If $\lim_{x \rightarrow \infty}  \lim_{t \rightarrow \infty} \inf {Pr}[Q_i^t < {x}] = 1$, the queue is \emph{substable}. If a queue is stable, then it is also substable. If a queue is not substable, then we say it is unstable.
\end{definition}

Loynes' theorem~\cite{b:Loynes} states that if the arrival and service processes of a queue are strictly jointly stationary and the average arrival rate is less than the average service rate, then the queue is stable.

\section{Performance Analysis for the Relay Queue}
\label{sec:analysis}
In this section, we derive expressions for key performance metrics for the relay queue, namely arrival and service rates, stability conditions, and average queue length. The analysis is provided for two cases: (i) when the network consists of two non-symmetric in general users, (ii) for $n>2$ symmetric users.

This section is an intermediate step before investigating the impact of the relay node in the per-user throughput, the aggregate throughput, and the average per packet delay.
In order to study those quantities, we need to first compute the average arrival and service rate of the relay, the average queue length, and the stability conditions.
The stability of a queue is translated to bounded queue size, which implies finite queuing delay.

\subsection{Two-user Case}
We study first the relay queue characteristics for the two-user case.
In this network, at each timeslot, the relay can receive at most two packets (one per user) and to transmit at most one.

The probability that the relay receives $i$ packets in a given timeslot when its queue is empty is denoted by $r_{i}^{0}$, and $r_{i}^{1}$ otherwise (not empty).
The expressions for the $r_{i}^{j}$ are rather lengthy and are presented in Appendix~\ref{sec:app_proof_2}.
The average arrival rate at the relay when its queue is empty is denoted by $\lambda_{0}$, and by $\lambda_{1}$ when it is not (derived in Appendix~\ref{sec:app_proof_2}).
The probability that the relay queue increases by $i$ packets when is empty is denoted by $p_{i}^{0}$, and $p_{i}^{1}$ when it is not; $p_{-1}^{i}$ is the probability
that the queue decreases by one packet. Note that $p_{i}^{j}$ and $r_{i}^{j}$ are in general different quantities, however $p_{i}^{j}=r_{i}^{j}$ in half-duplex relay systems\footnote{The case of half-duplex relay is studied in \cite{b:PappasMPR}, for which the analysis is simpler compared to the full-duplex case.}.

The next theorem presents the main relay queue characteristics for the two-user case.

\begin{thm} \label{thm:2users}
The key performance measures for the relay queue in a two-user network are provided below.
\begin{itemize}
\item[(i)]The average service rate is
\begin{equation} \label{eq:m2}
\begin{aligned}
\mu=q_{0}(1-q_{1})(1-q_{2})P_{0/0}^{d}+q_{0}q_{1}(1-q_{2})P_{0/0,1}^{d}+\\+q_{0}q_{2}(1-q_{1})P_{0/0,2}^{d}+q_{0}q_{1}q_{2}P_{0/0,1,2}^{d}.
\end{aligned}
\end{equation}
where $P_{0/0,i,j}^{d}$ is the success probability between the relay and the destination when the transmitting nodes are the relay and nodes $i$ and $j$. $P_{0/0,i,j}^{d}$ can be computed from (\ref{eq:succprob}).
\item[(ii)] The probability that the queue at the relay is empty is
\begin{equation}
\label{eq:probempty2}
P\left(Q=0\right)=\frac{p_{-1}^{1}-p_{1}^{1}-2p_{2}^{1}}{p_{-1}^{1}-p_{1}^{1}-2p_{2}^{1}+\lambda_{0}}.
\end{equation}
\item[(iii)]The average arrival rate $\lambda$ is
\begin{equation} \label{eq:lambda2}
\lambda=\frac{p_{-1}^{1}-p_{1}^{1}-2p_{2}^{1}}{p_{-1}^{1}-p_{1}^{1}-2p_{2}^{1}+\lambda_{0}}\lambda_{0}+\frac{\lambda_{0}}{p_{-1}^{1}-p_{1}^{1}-2p_{2}^{1}+\lambda_{0}}\lambda_{1}.
\end{equation}
\item[(iv)]The average relay queue size $Q$ is

{\scriptsize
\begin{equation}
\label{eq:avQ2}
\overline{Q}=\frac{(p_{1}^{1}+2p_{2}^{1}-p_{-1}^{1})(4p_{1}^{0}+10p_{2}^{0})+\lambda_{0}(2p_{-1}^{1}-4p_{1}^{1}-10p_{2}^{1})}{2(p_{1}^{1}+2p_{2}^{1}-p_{-1}^{1})(p_{-1}^{1}-p_{1}^{1}-2p_{2}^{1}+\lambda_{0})}.
\end{equation}
}

\end{itemize}
\end{thm}

\begin{proof}
See Appendix~\ref{sec:app_proof_2}.
\end{proof}

\emph{Note:} The values of $q_{0}$ for which the queue is stable are given by $q_{0min}<q_{0}<1$, where $q_{0min}$ is given in (\ref{eq:q0min2}) in Appendix~\ref{sec:app_proof_2}.
Queue stability is an important parameter of quality-of-service (QoS), as it implies finite queue delay (due to bounded queue size). The queueing delay is computed in Section~\ref{sec:delay_analysis}.

\subsection{Symmetric $n$-user Case} \label{sec:relay-results_n}
We now investigate the case of a symmetric $n$-user network\footnote{Our work could be generalized to the asymmetric case; nevertheless the expressions will be significantly involved  without providing any meaningful or crisp insights.}. Each user attempts to transmit in a slot with probability $q$; the success probability to the relay and the destination when $i$ nodes transmit are given by $P_{0,i}$ and $P_{d,i}$, respectively. There are two cases for the $P_{d,i}$, i.e., $P_{d,i,0}$ and $P_{d,i,1}$, denoting the success probability when relay remains silent or transmits, respectively. Those success probabilities for the symmetric case are given by $P_{d,i,j}=P_{d}\left(\frac{1}{1+\gamma_{d}} \right)^{i-1} \left(\frac{1}{1+\beta\gamma_{0}} \right)^{j},\text{ } j=0,1$ and $\beta=\frac{v_{0d}h_{0d}}{v_{d}h_{d}}>1$. $P_{0d,i}=P_{0d}\left(\frac{1}{1+\frac{1}{\beta}\gamma_{d}}\right)^{i}$, $P_{0}=\exp\left(-\frac{\gamma_{0}\eta_{0}}{v_{0}h_{0}}\right)$, $P_{d}=\exp\left(-\frac{\gamma_{d}\eta_{d}}{v_{d}h_{d}}\right)$, $P_{0d}=\exp\left(-\frac{\gamma_{0}\eta_{0}}{v_{0}h_{0}}\right)$. There are two cases for the $P_{0,i}$, i.e., $P_{0,i,0}$ and $P_{0,i,1}$, denoting the success probability when the relay remains silent or transmits respectively. The success probabilities are given by $P_{0,i,0}=P_{0}\left(\frac{1}{1+\gamma_{0}} \right)^{i-1}$ and $P_{0,i,1}=P_{0}\left(1+\gamma_{0}r_{0}^{\alpha}g \right)^{-1}\left(\frac{1}{1+\gamma_{0}} \right)^{i-1}$, where $r_{0}$ is the distance between the users and the relay, $v_i$ is the parameter of the Rayleigh fading random variable at channel $i$, $\alpha$ is the path loss exponent and $g$ is the self-interference coefficient. 

The next theorem summarizes the results for the characteristics of the relay queue for the symmetric $n$-user case.

\begin{thm} \label{thm:nusers}
The key performance measures for the relay queue in the $n$-symmetric user network are provided below.
\begin{itemize}
\item[(i)] The average service rate is
\begin{equation}
\label{eq:mun}
\mu=\sum_{k=0}^{n}{{n \choose k} {q_{0}q^{k}(1-q)^{n-k}}P_{0d,k}}.
\end{equation}
\item[(ii)] The probability that the queue at the relay is empty is
\begin{equation} \label{eq:probemptyn}
P\left( Q=0 \right)=\frac{\displaystyle p_{-1}^{1}-\sum_{i=1}^{n}{ip_{i}^{1}}}{\displaystyle p_{-1}^{1}-\sum_{i=1}^{n}{ip_{i}^{1}}+\lambda_{0}}.
\end{equation}
\item[(iii)]The average arrival rate $\lambda$ is
\begin{equation}
\label{eq:lamdan}
\lambda=P\left(Q=0\right)\lambda_{0}+P\left(Q>0\right)\lambda_{1}.
\end{equation}
The expressions for $\lambda_{0}$ and $\lambda_{1}$ are given in Appendix~\ref{sec:app_proof_n}.
\item[(iv)]The average relay queue size $Q$ is

\begin{equation}
{\scriptsize
\label{eq:avQn}
\begin{aligned}
\overline{Q}=\frac{\displaystyle \left(\sum_{i=1}^{n}{ip_{i}^{1}}-p_{-1}^{1} \right)\sum_{i=1}^{n}{i(i+3)p_{i}^{0}}+\lambda_{0}\left(2p_{-1}^{1}-\sum_{i=1}^{n}{i(i+3)p_{i}^{1}} \right)}{\displaystyle 2\left(\sum_{i=1}^{n}{ip_{i}^{1}}-p_{-1}^{1} \right) \left(p_{-1}^{1}-\sum_{i=1}^{n}{ip_{i}^{1}}+\lambda_{0} \right)}.
\end{aligned}
}
\end{equation}
\end{itemize}
\end{thm}

\begin{proof}
See Appendix~\ref{sec:app_proof_n}.
\end{proof}

The values of $q_{0}$ for which the queue is stable are given by $q_{0min}<q_{0}<1$, where $q_{0min}$ is given in (\ref{eq:q0minn}) in Appendix~\ref{sec:app_proof_n}.

\section{Throughput Analysis} \label{sec:thr_analysis}
In the previous section, we provided the main results on the relay queue characteristics, including the empty queue probability and the average queue length. Here, we derive the per-user throughput and the network aggregate throughput with one cooperative relay and $n$ users.

The per-user throughput, $T_i$ for the $i$-th user is given by $T_i=T_{D,i}+T_{R,i}$, where $T_{D,i}$ denotes the direct throughput from user $i$ to the destination, i.e., the transmitted packet reaches the destination directly, without using the relay. When the transmission to the destination is not successful, and at the same time the relay node receives the packet correctly, then it stores it to its queue, and the contributed throughput by the relay for the user $i$ is denoted by $T_{R,i}$.
When the queue at the relay is stable, $T_{R,i}$ is the arrival rate from user $i$ to the queue.

The term $T_{D,i}$ can also be interpreted as the probability that a transmitted packet from user $i$ reaches the destination directly, and $T_{R,i}$ is the probability of unsuccessful transmission from user $i$ to the destination while the packet is received at the relay.

The percentage of $i$-th user's traffic that is being relayed is $\frac{T_{R,i}}{T_i}$.

In the following subsection, we provide expressions for $T_{D,i}$ and $T_{R,i}$ for the two-user and the symmetric $n$-user cases.

\subsection{Per-user and Aggregate Throughput: Two-user Case}
The direct throughput to the destination for the $i$-th user, $T_{D,i}$, is given by
\begin{equation}
\begin{aligned}
T_{D,i} = q_{0}P\left(Q>0\right)q_{i} \left[  (1-q_{j}) P_{i/0,i}^{d}+q_{j} P_{i/0,i,j}^{d} \right]+ \\
+\left[1-q_{0}P\left(Q>0\right)\right]q_{i} \left[ (1-q_{j}) P_{i/i}^{d}+q_{j} P_{i/i,j}^{d} \right].
\end{aligned}
\end{equation}

When the relay queue is stable, the contributed throughput to user $i$, $T_{R,i}$, is the arrival rate from user $i$ to the relay queue.
Note that a packet from user $i$ enters the relay queue when the transmission to the destination is not successful and at the same time the relay is
able to decode that packet. The relayed throughput $T_{R,i}$ of user $i$ is given by

\begin{equation}
{\scriptsize
\begin{aligned}
T_{R,i}=q_{0}P\left(Q>0\right)q_{i} \left[  (1-q_{j}) (1-P_{i/0,i}^{d})P_{i/0,i}^{0} +q_{j} (1-P_{i/0,i,j}^{d})P_{i/0,i,j}^{0} \right]+\\
+\left[1-q_{0}P\left(Q>0\right)\right]q_{i} \left[ (1-q_{j}) (1-P_{i/i}^{d})P_{i/i}^{0} +q_{j} (1-P_{i/i,j}^{d})P_{i/i,j}^{0}\right].
\end{aligned}
}
\end{equation}

The throughput $T_{i}$ for the $i$-th user is given by
\begin{equation}
\label{eq:thr2_1}
T_{i}= T_{D,i} + T_{R,i}.
\end{equation}

In the above equations, the queue is assumed to be stable, hence the arrival rate from each user to the queue is a contribution to the overall throughput. The aggregate throughput is $T_{aggr} = T_{1} + T_{2}$. Notice that the per-user throughput is independent of $q_{0}$ as long as it is in the stability region. This is due to the fact that the product $q_{0}P\left(Q>0\right)$ is constant (does not depend on $q_0$). The proof is straightforward and thus is omitted.

\subsection{Per-user and Aggregate Throughput: Symmetric $n$-user Case}
In this subsection, we provide expressions for the direct and the relayed per-user and aggregate throughput.
The notation used in Section~\ref{sec:relay-results_n} applies here as well. Furthermore, the per-user throughput is denoted by $T$, the direct throughput to the destination by $T_D$, and the relayed throughput by $T_R$.

The direct throughput $T_{D}$ is given by
\begin{equation}
\begin{aligned}
T_{D}=q_{0}P\left(Q>0\right)\sum_{k=0}^{n-1}{{n-1 \choose k}q^{k+1}(1-q)^{n-1-k} P_{d,k+1,1}}+\\
+\left[1-q_{0}P\left(Q>0\right)\right]\sum_{k=0}^{n-1}{{n-1 \choose k}q^{k+1}(1-q)^{n-1-k} P_{d,k+1,0}}.
\end{aligned}
\end{equation}

The throughput contributed by the relay (when the queue at the relay is stable), $T_{R}$, is given by
\begin{equation}
{\scriptsize
\begin{aligned}
T_{R}=q_{0}P\left(Q>0\right)\sum_{k=0}^{n-1}{{n-1 \choose k}q^{k+1}(1-q)^{n-1-k} (1-P_{d,k+1,1})P_{0,k+1,1}}+ \\
+\left[1-q_{0}P\left(Q>0\right)\right]\sum_{k=0}^{n-1}{{n-1 \choose k}q^{k+1}(1-q)^{n-1-k} \left(1-P_{d,k+1,0}\right)P_{0,k+1,0}}.
\end{aligned}
}
\end{equation}

The per-user throughput $T$ for the cooperative relay network when the relay queue is stable is given by
\begin{equation}
\label{eq:thrn}
T = T_D + T_R.
\end{equation}

The aggregate throughput is $T_{aggr} = n T$.

\begin{remark}
When the queue is unstable, the aggregate throughput is the summation of the direct throughput among the users and the destination plus the service rate of the relay. However, when the queue is unstable, the queue size increases to infinity, thus there is no guarantee for finite queueing delay.
\end{remark}

\section{Delay Analysis} \label{sec:delay_analysis}
In Section~\ref{sec:analysis}, we studied the performance of the relay queue in terms of the probability of empty queue and the average queue length.
That section was an intermediate step for our main goal, which is to study the impact of the relay node in the network in terms of throughput and the delay. In the previous section, we obtained the per-user and the aggregate throughput for a relay network with stable relay queue and commented on the case of unstable relay queue.
In this section, we analyze an important network performance measure, the delay, and derive analytical expressions for the average delay required to deliver a packet from the source to the destination.

\begin{thm}
\label{thm:delay}
The average delay for a packet received at the destination when it is in the head of the user queue is given by
\begin{equation}
\label{eq:delay}
D_i=\frac{1+T_{R,i}\left(\frac{\overline{Q}}{\lambda} + \frac{1}{\mu} \right)}{T_i},
\end{equation}
\end{thm}
where $T_{R,i}$ and $T_i$ is the $i$-th user relayed and per-user throughput, respectively. $\lambda$ and $\mu$ is the average arrival and service rate of the relay, respectively, and $\overline{Q}$ is the average queue length of the relay.
\begin{proof}
See Appendix~\ref{sec:app_delay}.
\end{proof}

The expressions for $T_{R,i}$ and $T_i$ are given in Section~\ref{sec:thr_analysis}. The expressions for $\lambda$, $\mu$, $\overline{Q}$  are summarized in Theorem~\ref{thm:2users} and~\ref{thm:nusers} for the two-user and the symmetric $n$-user case, respectively.

Note that the term $\frac{\overline{Q}}{\lambda}$ in (\ref{eq:delay}) is the queueing delay, which is the time a packet spends in queue,
the time the packet is assigned to the queue for transmission and the time it starts being transmitted. In the meantime, the packet waits while other packets in the queue are transmitted.

\begin{remark}
When the relay queue is unstable, the average queue length can be arbitrarily large, thus the average queueing delay
tends to infinity. In (\ref{eq:delay}), when the queue is unstable, then the average delay also tends to infinity. In the case of unstable queues, flow control policies could be applied for packet dropping, however this is beyond the scope of our paper.
\end{remark}

\section{Numerical Results}
\label{sec:results}
In this section, we provide numerical results to validate the above theoretical performance analysis.
For exposition convenience, we consider the case where all users have the same link characteristics and transmission probabilities. The parameters used in the numerical results are as follows: distances are $r_{d}=130$, $r_{0}=60$, and $r_{0d}=80$ in meters, the path loss exponent is $\alpha=4$, and the receiver noise power $\eta=10^{-11}$. The transmit power for the relay is $P_{tx}(0)=10$ mW and for the $i$-th user is $P_{tx}(i)=1$ mW.

\subsection{Per user and Aggregate Throughput}
Figs.~\ref{fig:thr_n_02} and~\ref{fig:thr_n_06} present the per-user throughput versus the number of users in the network for different values of $q$ and $g$, and for $\gamma=0.2$ and $\gamma=0.6$, respectively. Figs.~\ref{fig:athr_n_02} and~\ref{fig:athr_n_06} show the aggregate throughput versus the number of users.
When $\gamma=0.2$, we observe that for $g=10^{-10}$ and $g=10^{-8}$ (almost perfect self-interference cancelation) the relay queue is unstable for relative small number of users. This is because for small values of $\gamma$, it is more likely to have more successful transmissions from the users to the relay, while at the same time the relay can transmit at most one packet per timeslot. For $\gamma=0.6$ the queue is never unstable for the selected set of parameters, while for $g=10^{-10}$ and $g=10^{-8}$, throughput gains are evident as compared to no self-interference cancelation.

\begin{figure}[h!]
\centering
\subfigure[Per-user throughput vs. the number of users.]{
\includegraphics[scale=0.6]{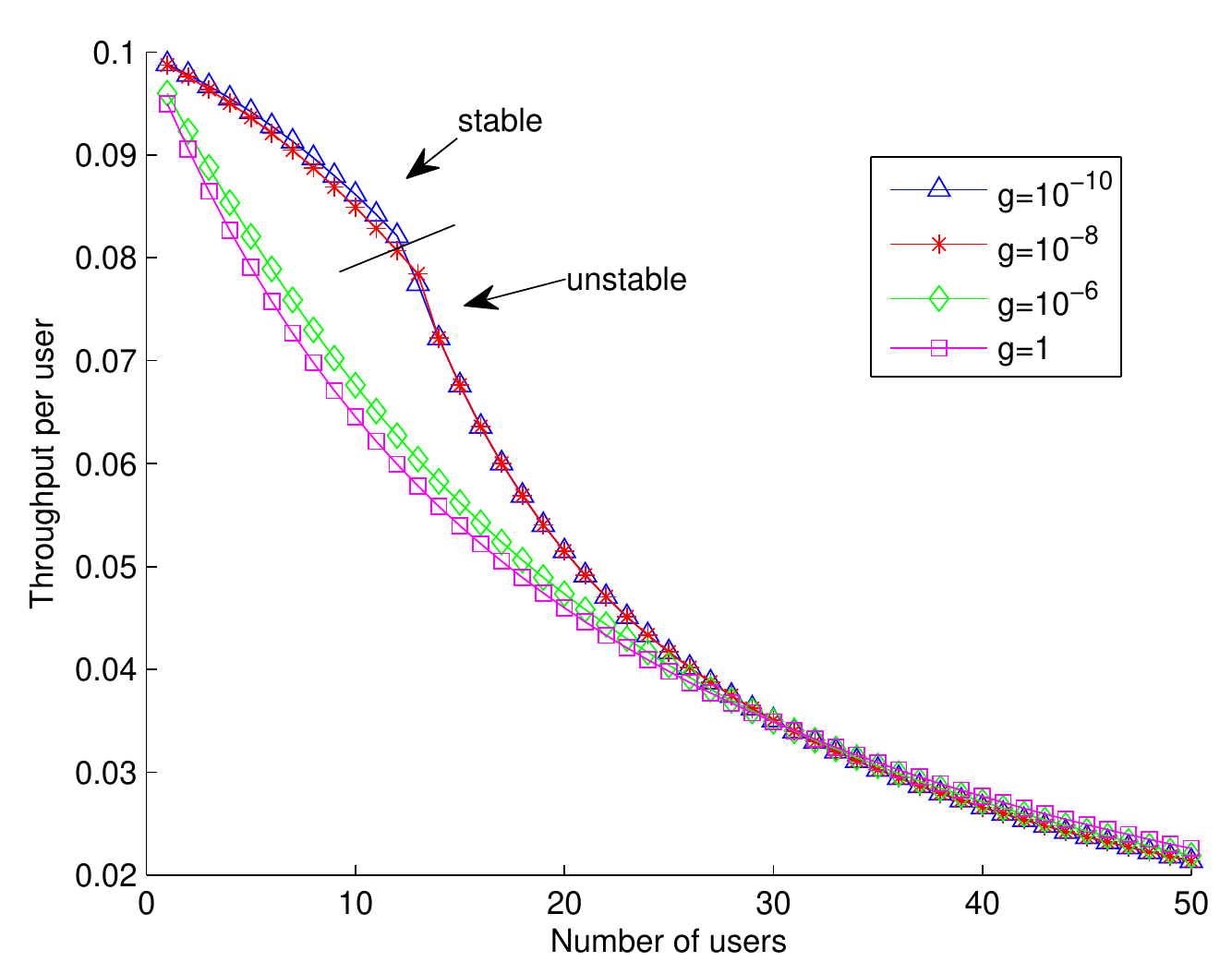}
\label{fig:thr_n_02}
}
\subfigure[Aggregate throughput vs. the number of users.]{
\includegraphics[scale=0.6]{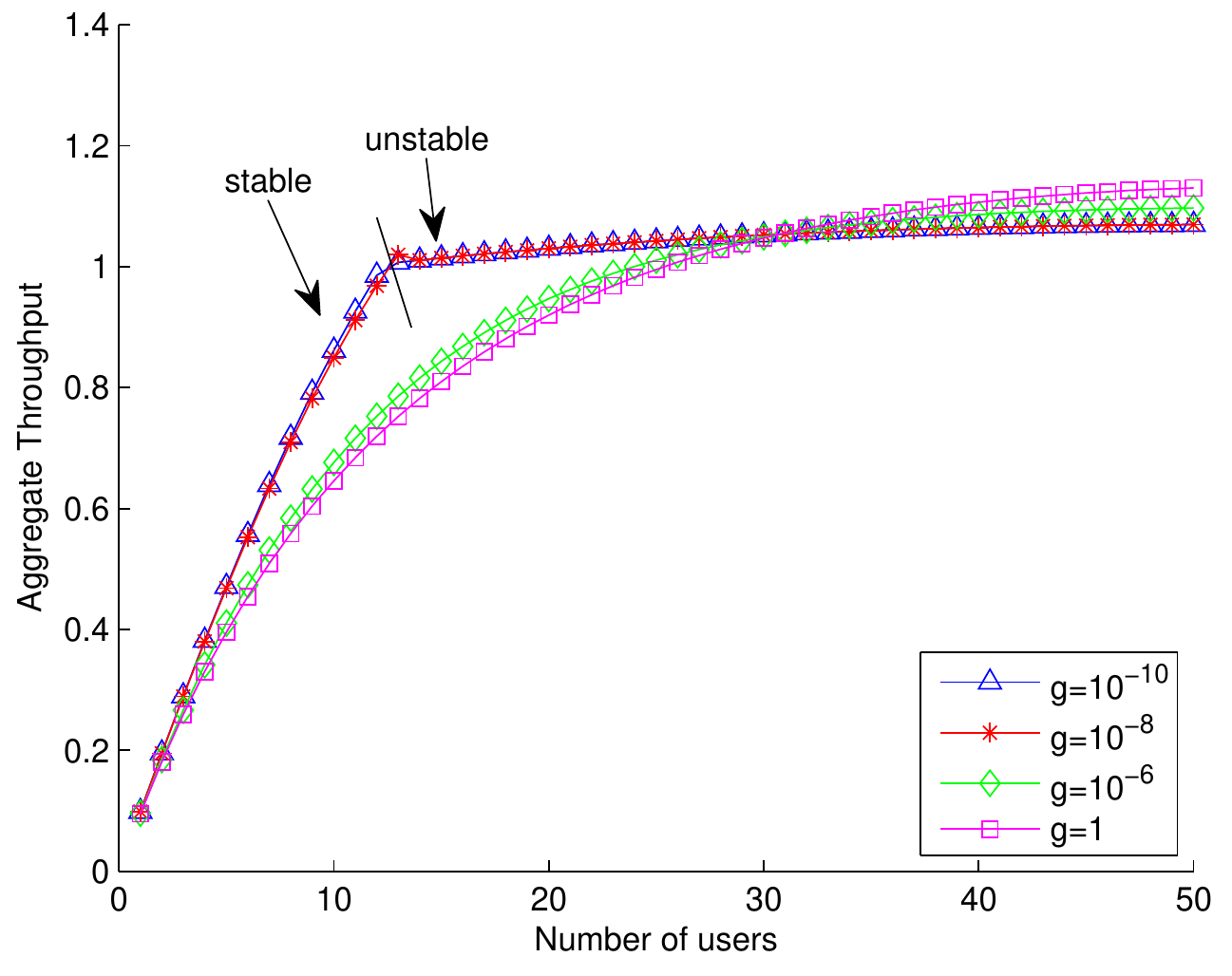}
\label{fig:athr_n_02}
}
\caption{Per-user and aggregate throughput vs. the number of users for $\gamma=0.2$, $q=0.1$ and $q_0 = 0.95$.}
\end{figure}

In Figs.~\ref{fig:percent_02} and~\ref{fig:percent_06}, we plot the percentage of traffic that is being relayed in the network (cf. Section~\ref{sec:thr_analysis}) for $\gamma=0.2$ and $\gamma=0.6$ respectively, for the case of a stable queue.

\begin{figure}[h!]
\centering
\includegraphics[scale=0.6]{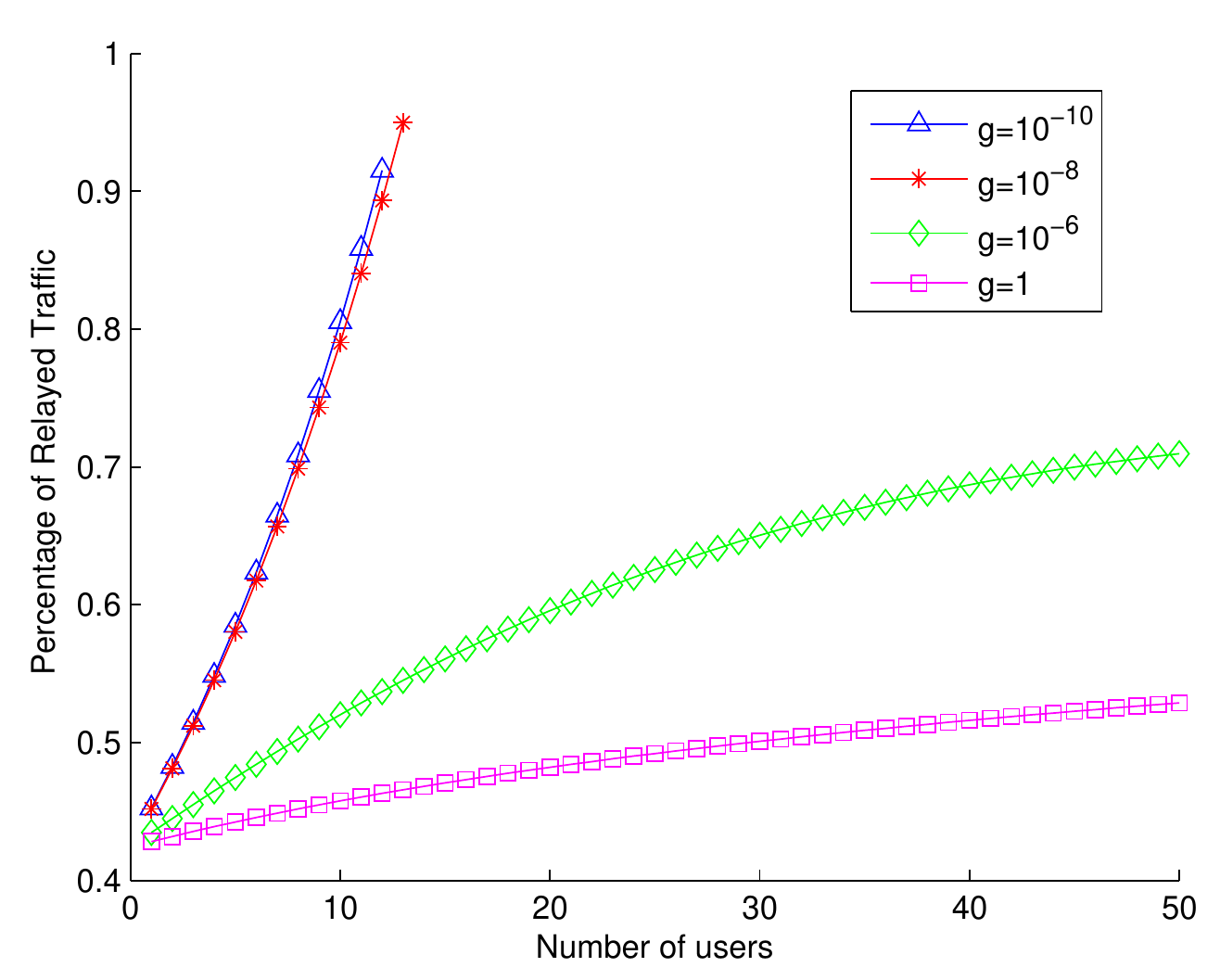}
\caption{Percentage of traffic that is being relayed vs. the number of users $\gamma=0.2$, $q=0.1$, and $q_0 = 0.95$.}
\centering
\label{fig:percent_02}
\end{figure}

\begin{figure}[h!]
\centering
\subfigure[Per-user throughput vs. the number of users.]{
\includegraphics[scale=0.6]{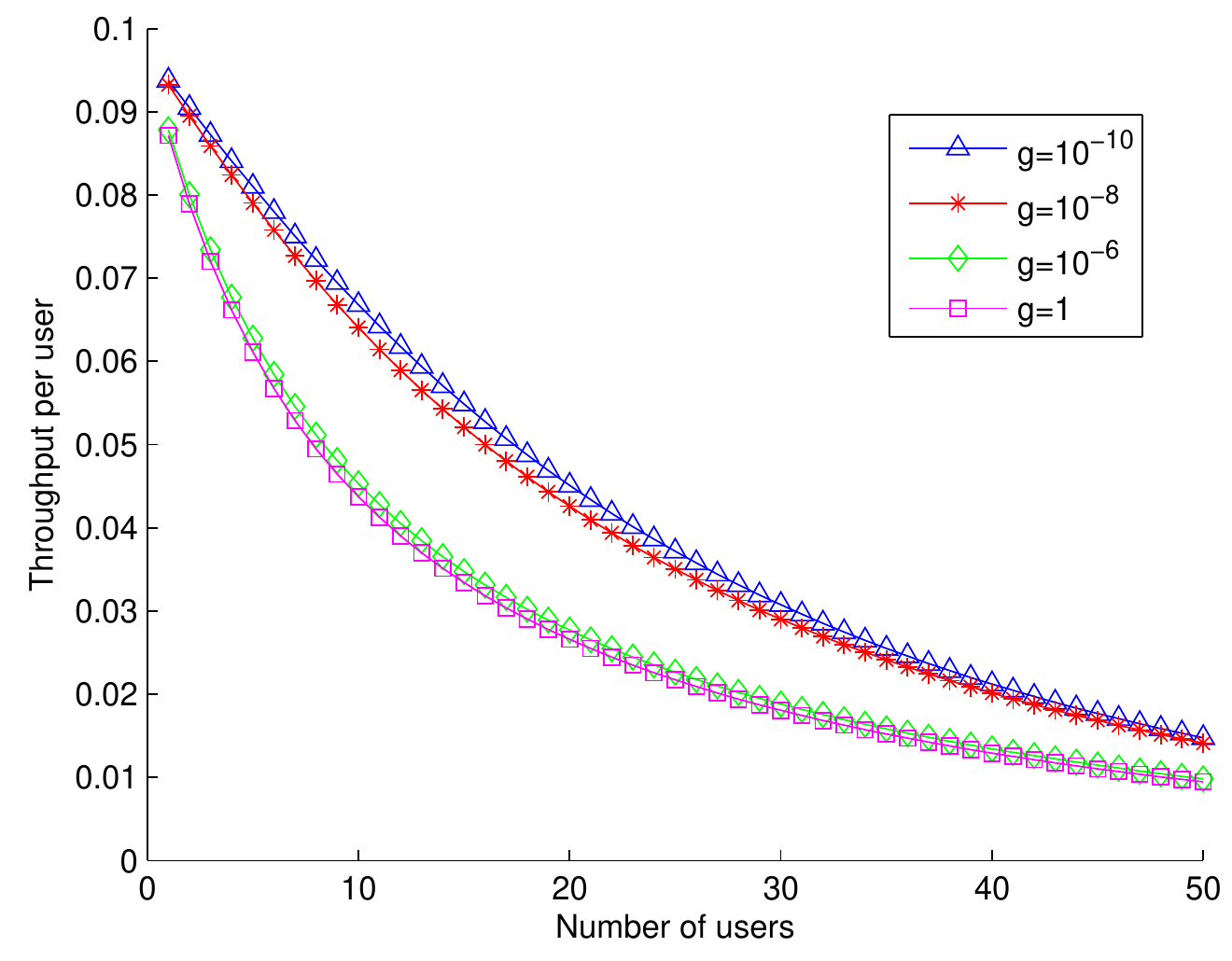}
\label{fig:thr_n_06}
}
\subfigure[Aggregate throughput vs. the number of users.]{
\includegraphics[scale=0.6]{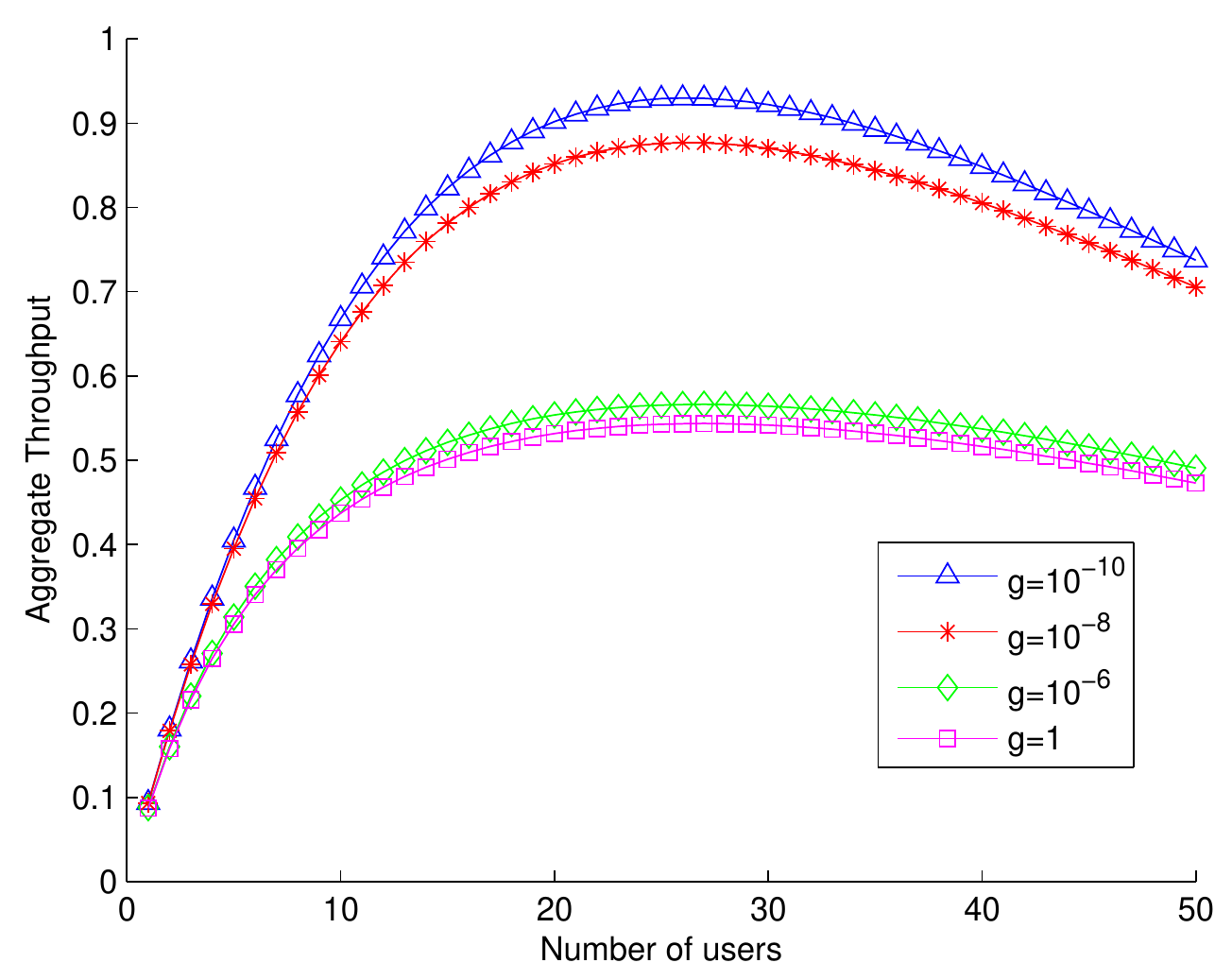}
\label{fig:athr_n_06}
}
\caption{Per-user and aggregate throughput vs. the number of users for $\gamma=0.6$, $q=0.1$, and $q_0 = 0.99$.}
\end{figure}

\begin{figure}[h!]
\centering
\includegraphics[scale=0.6]{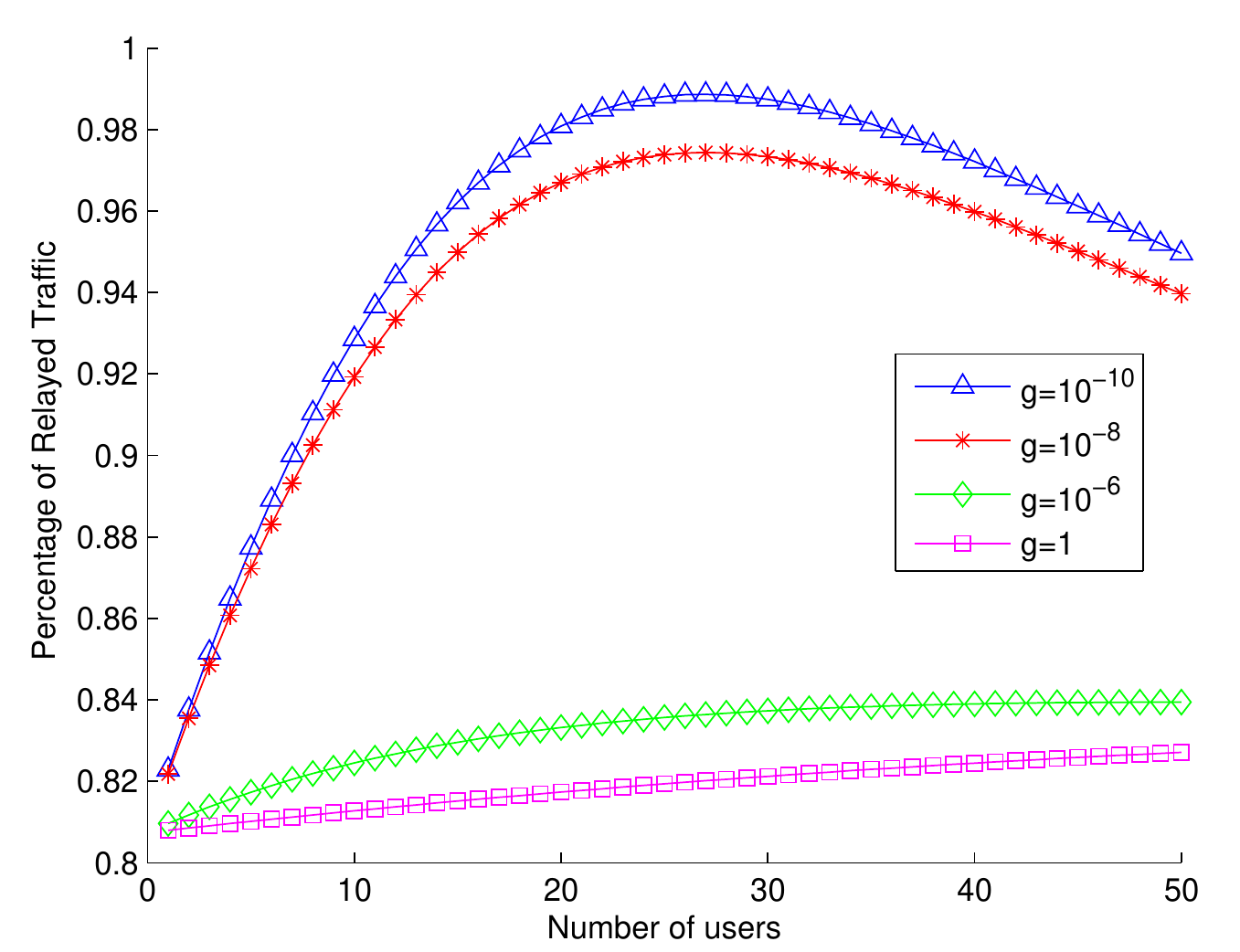}
\caption{Percentage of traffic that is being relayed vs. the number of users $\gamma=0.6$, $q=0.1$, and $q_0 = 0.99$.}
\centering
\label{fig:percent_06}
\end{figure}

Figs.~\ref{fig:thr_n_12} and~\ref{fig:thr_n_25} present the per-user throughput versus the number of users in the network for different values of $q$ and $g$, and for $\gamma=1.2$ and $\gamma=2.5$, respectively. Figs.~\ref{fig:athr_n_12} and~\ref{fig:athr_n_25} show the aggregate throughput versus the number of users. Finally, Figs.~\ref{fig:percent_12}
and~\ref{fig:percent_25} show the percentage of traffic that is being relayed.

Note that when the percentage tends to $1$ (or $100\%$), the contributed throughput by the relay tends to be the total network throughput.

The gains from the relay are more pronounced for large $\gamma$, whilst in the case of $\gamma=0.2$ and quasi perfect self-interference cancelation, we tend to have an unstable queue, which affects the delay per packet as we will see in the next subsection.

\begin{figure}[h!]
\centering
\subfigure[Per-user throughput vs. the number of users.]{
\includegraphics[scale=0.6]{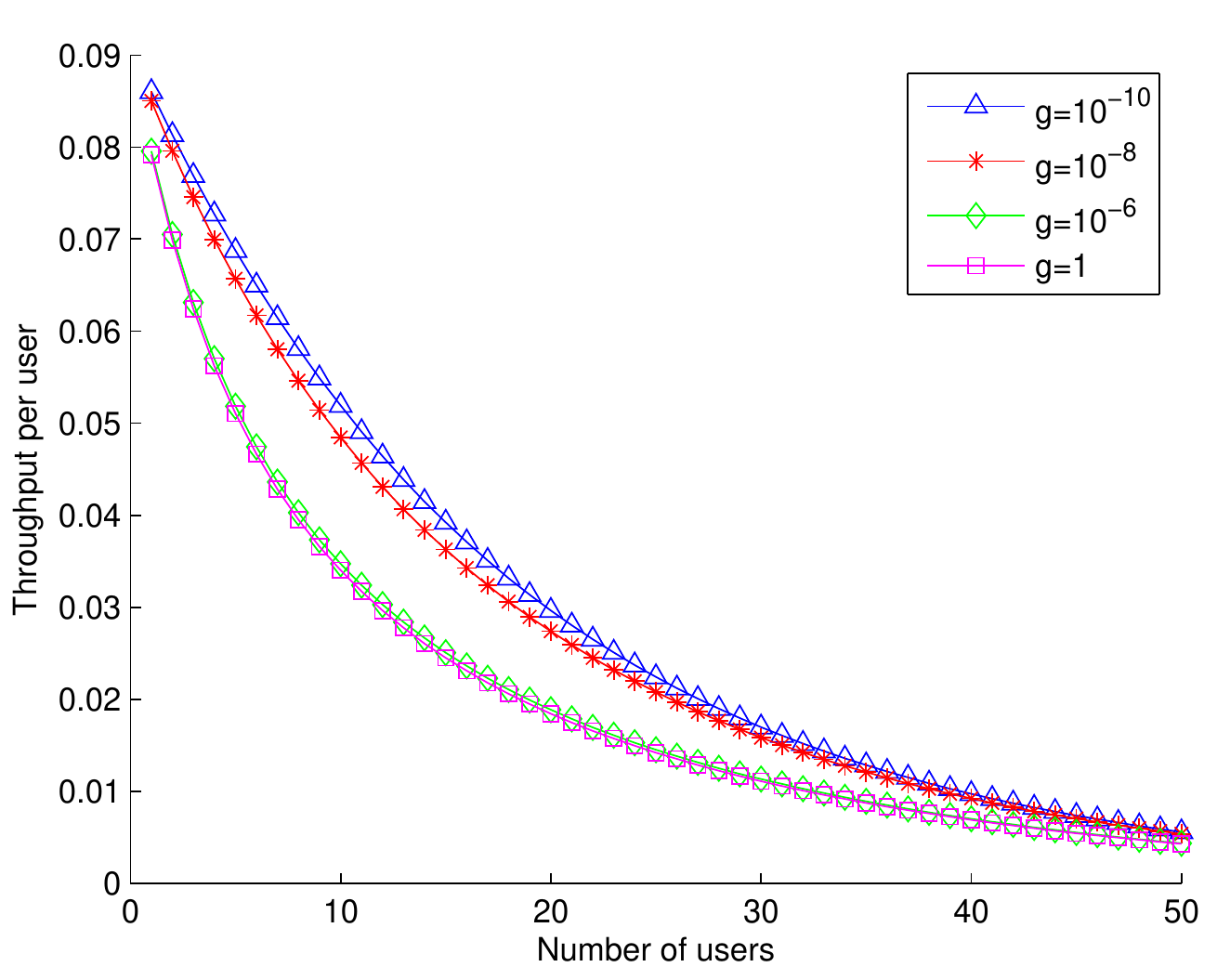}
\label{fig:thr_n_12}
}
\subfigure[Aggregate throughput vs. the number of users.]{
\includegraphics[scale=0.6]{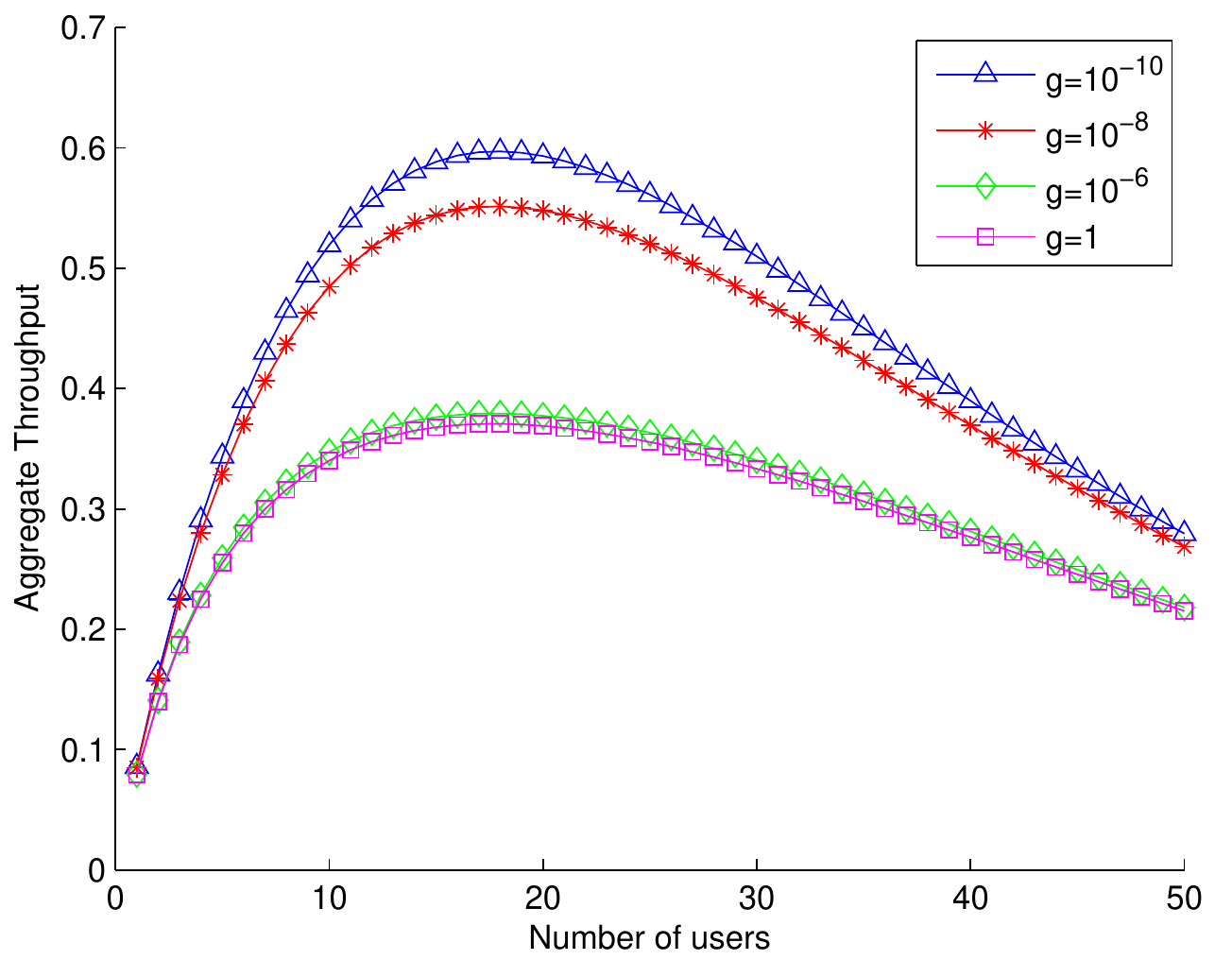}
\label{fig:athr_n_12}
}
\caption{Per-user and aggregate throughput vs. the number of users for $\gamma=1.2$, $q=0.1$, and $q_0 = 0.99$.}
\end{figure}

\begin{figure}[h!]
\centering
\includegraphics[scale=0.6]{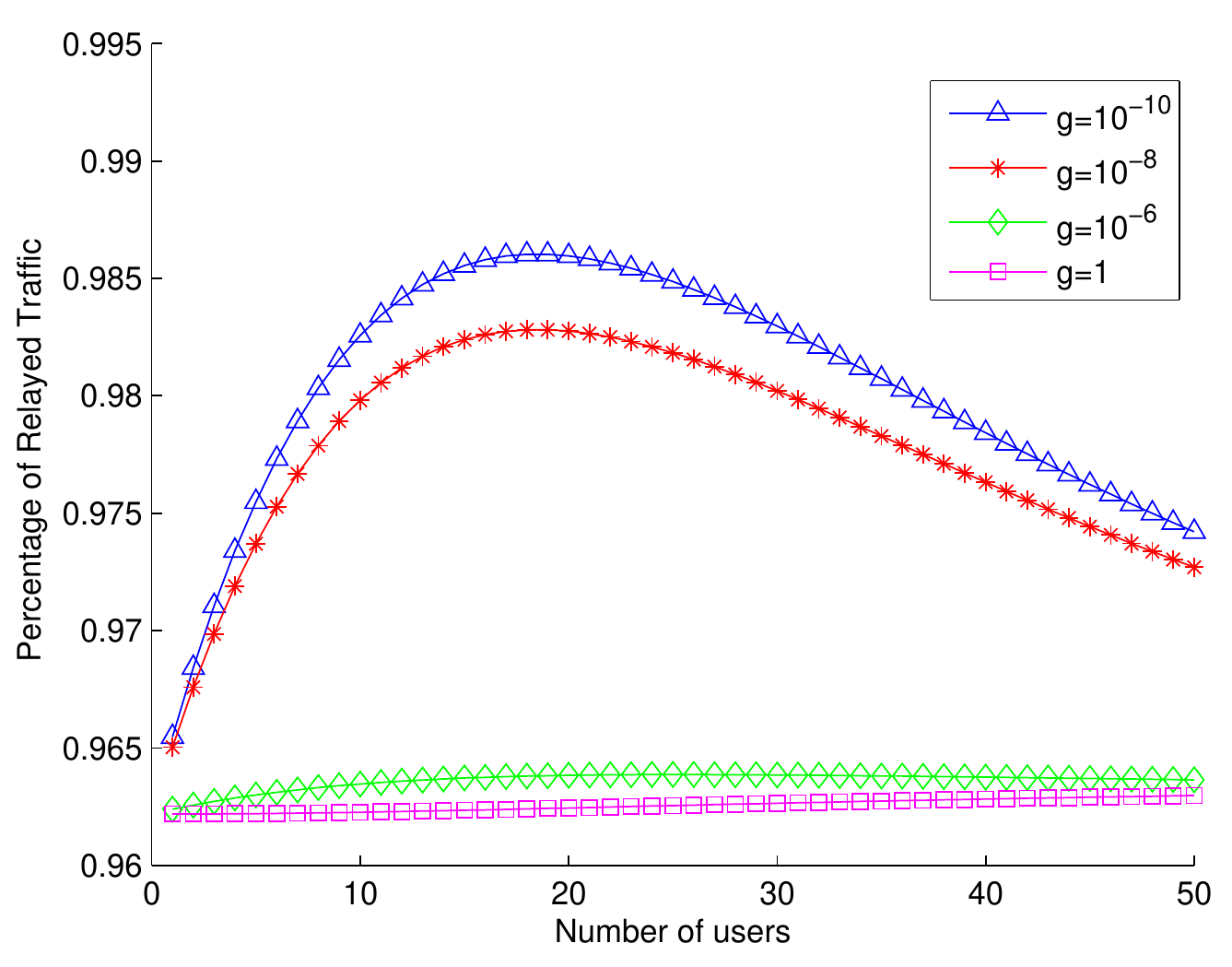}
\caption{Percentage of traffic that is being relayed vs. the number of users $\gamma=1.2$, $q=0.1$, and $q_0 = 0.99$.}
\centering
\label{fig:percent_12}
\end{figure}

\begin{figure}[h!]
\centering
\subfigure[Per-user throughput vs. the number of users.]{
\includegraphics[scale=0.6]{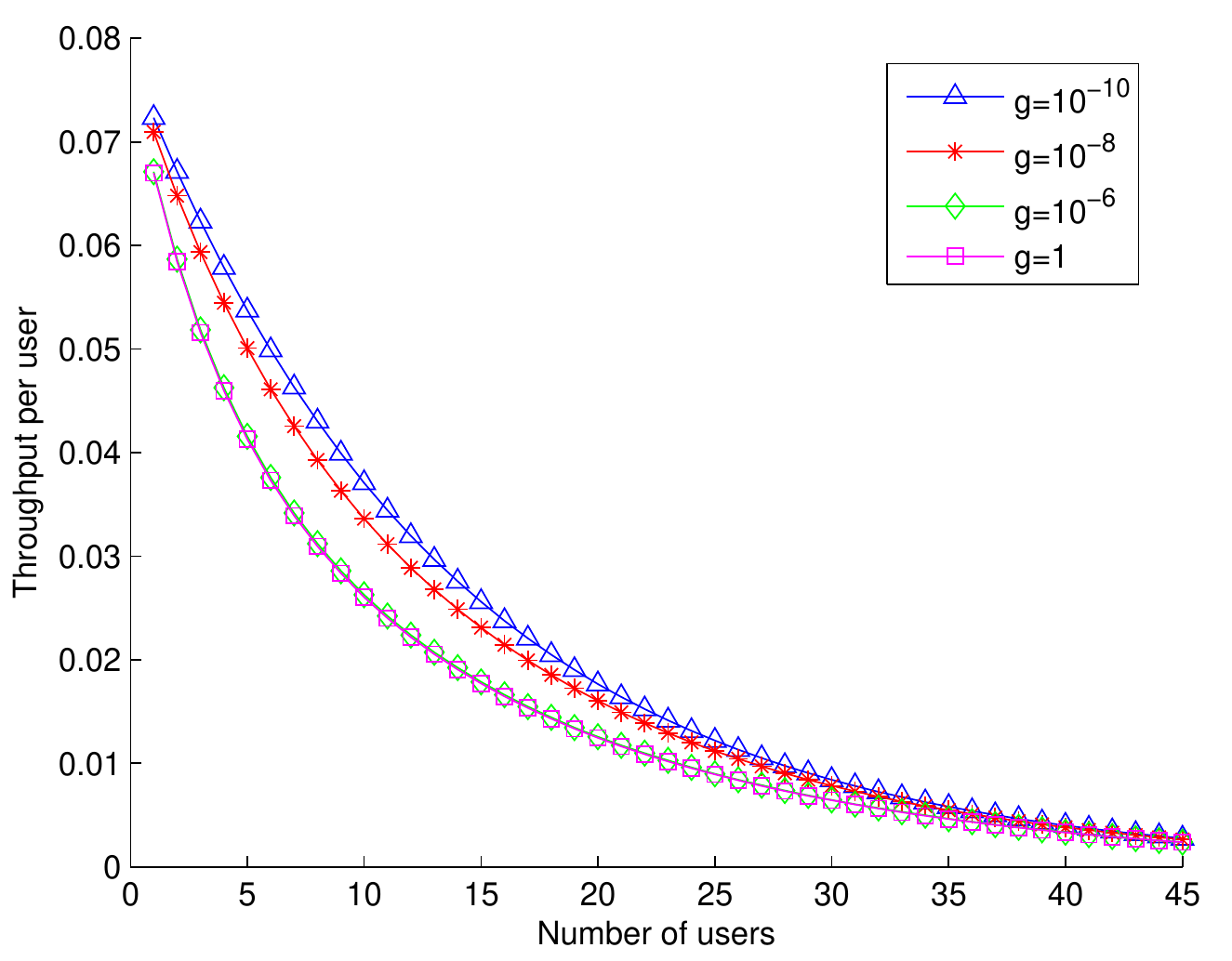}
\label{fig:thr_n_25}
}
\subfigure[Aggregate throughput vs. the number of users.]{
\includegraphics[scale=0.6]{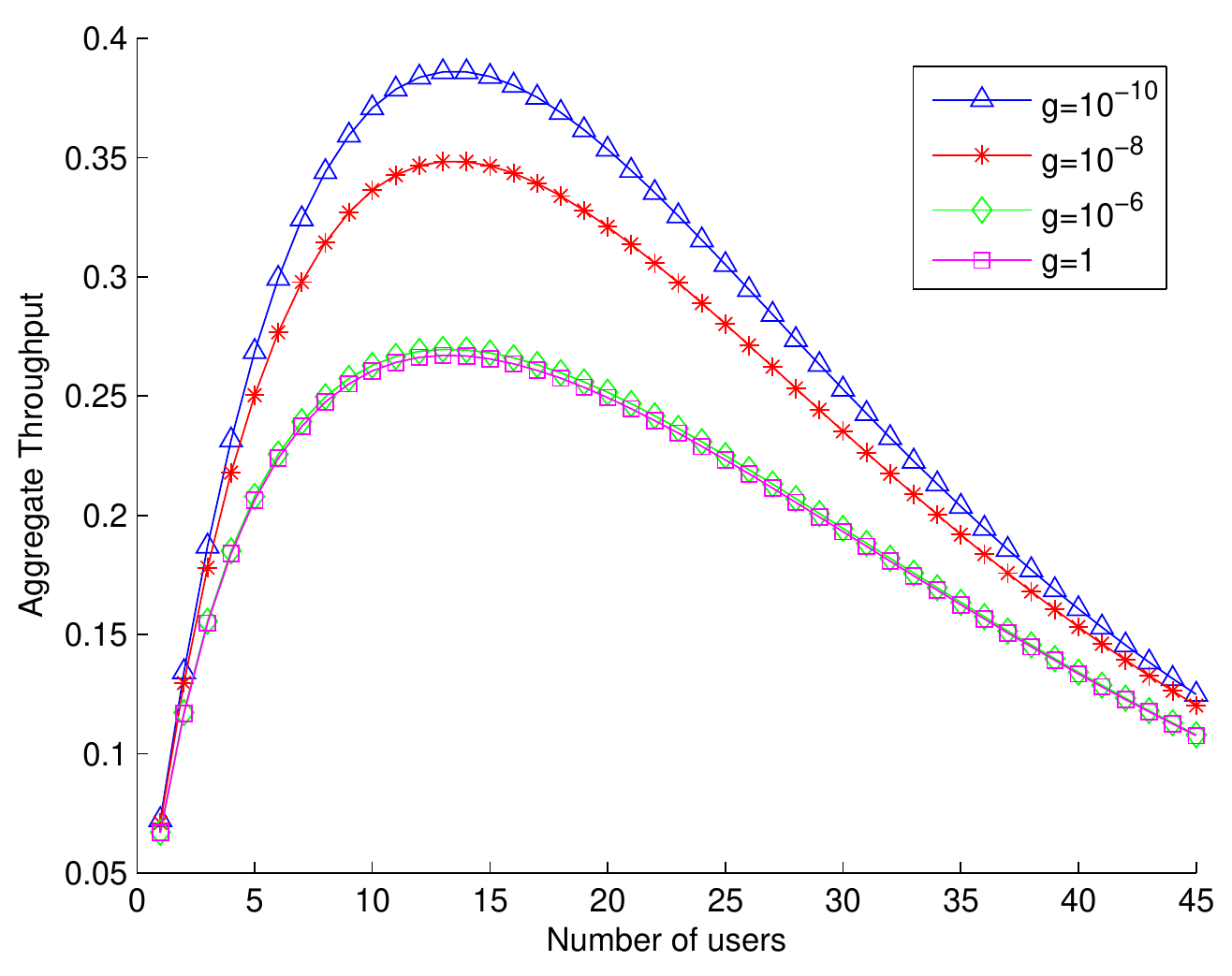}
\label{fig:athr_n_25}
}
\caption{Per-user and aggregate throughput vs. the number of users for $\gamma=2.5$, $q=0.1$, and $q_0 = 0.99$.}
\end{figure}

\begin{figure}[h!]
\centering
\includegraphics[scale=0.6]{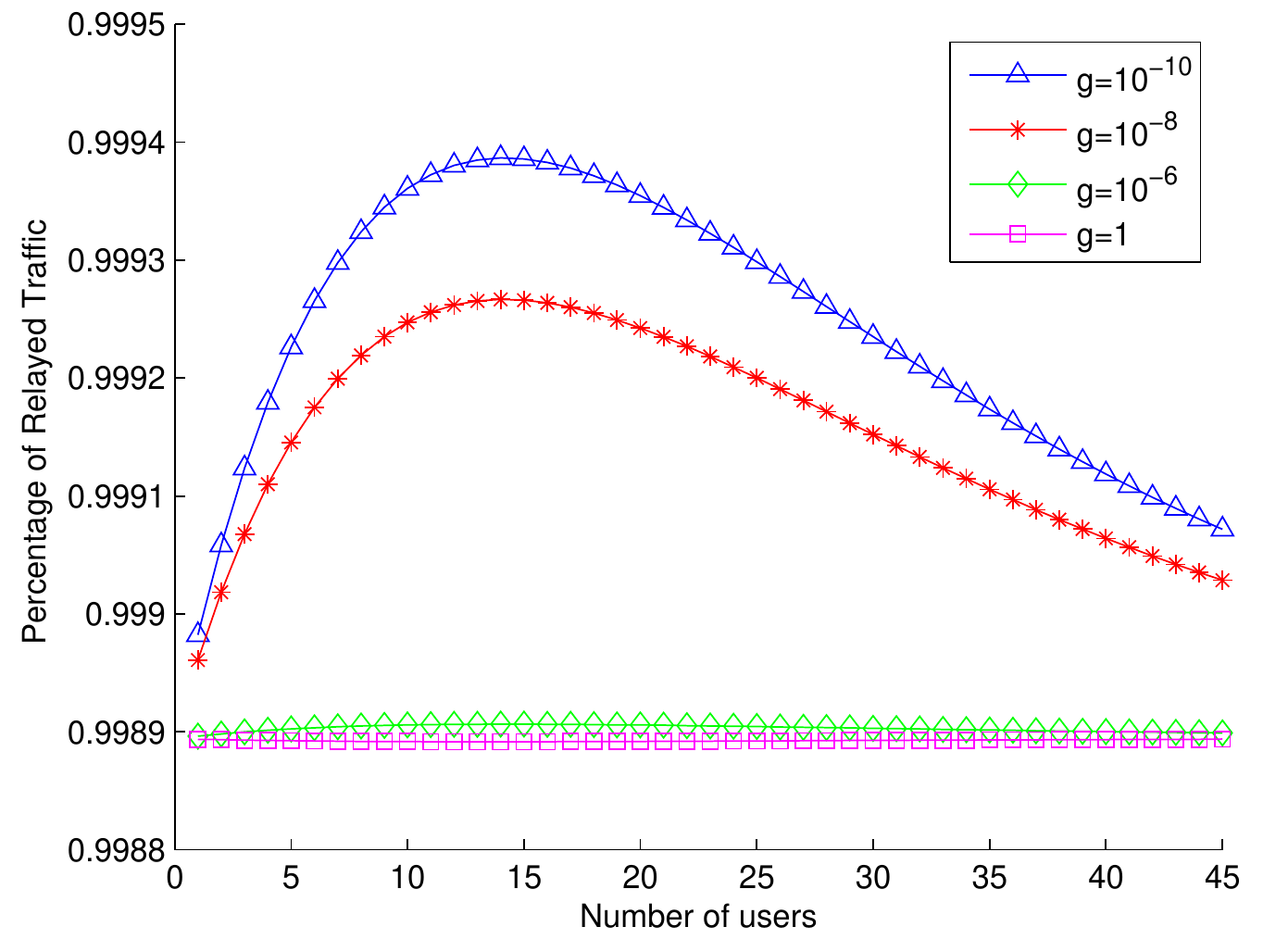}
\caption{Percentage of traffic that is being relayed vs. the number of users $\gamma=2.5$, $q=0.1$ and $q_0 = 0.99$.}
\centering
\label{fig:percent_25}
\end{figure}

\subsection{Average Queue Length and Average Delay per Packet}
In this subsection, we provide numerical results for two key performance metrics, namely the average relay queue size and the average delay per packet.

Figs.~\ref{fig:avQ_n_02} and~\ref{fig:avQ_n_06} present the average queue length of the relay for $\gamma=0.2$ and $\gamma=0.6$. The average queue length is among the factors that affect the average delay per packet as presented in (\ref{eq:delay}) of Theorem~\ref{thm:delay}. Figs.~\ref{fig:Delay_n_02} and~\ref{fig:Delay_n_06} illustrate the average delay per packet for $\gamma=0.2$ and $\gamma=0.6$.

For $\gamma=0.2$ we included the per-packet delay for the network without the relay for comparison reasons. We observe that in that case the cooperative relay node does not provide any gains for increasing number of users, as the delay for the relay network is larger than the delay without using a relay.
For $\gamma=0.6$, the delay for the network without the relay is much larger, e.g., it starts with $50$ timeslots for $1$ users and goes up to $400$ for $50$ users. In that case, the use of a relay is beneficial in terms of throughput and per-packet delay.

\begin{figure}[h!]
\centering
\subfigure[Average queue length vs. the number of users]{
\includegraphics[scale=0.6]{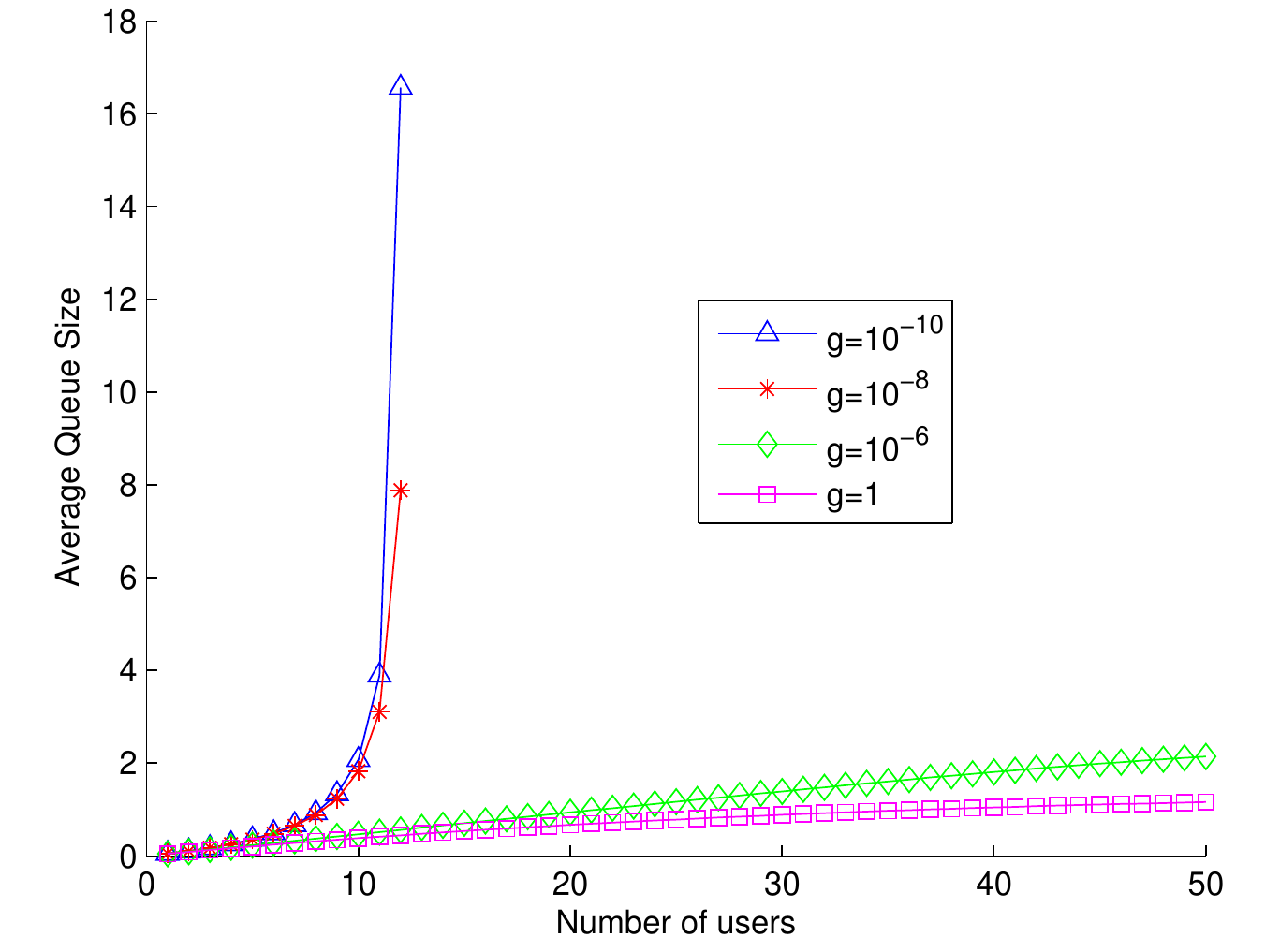}
\label{fig:avQ_n_02}
}
\subfigure[Per-packet delay vs. the number of users]{
\includegraphics[scale=0.6]{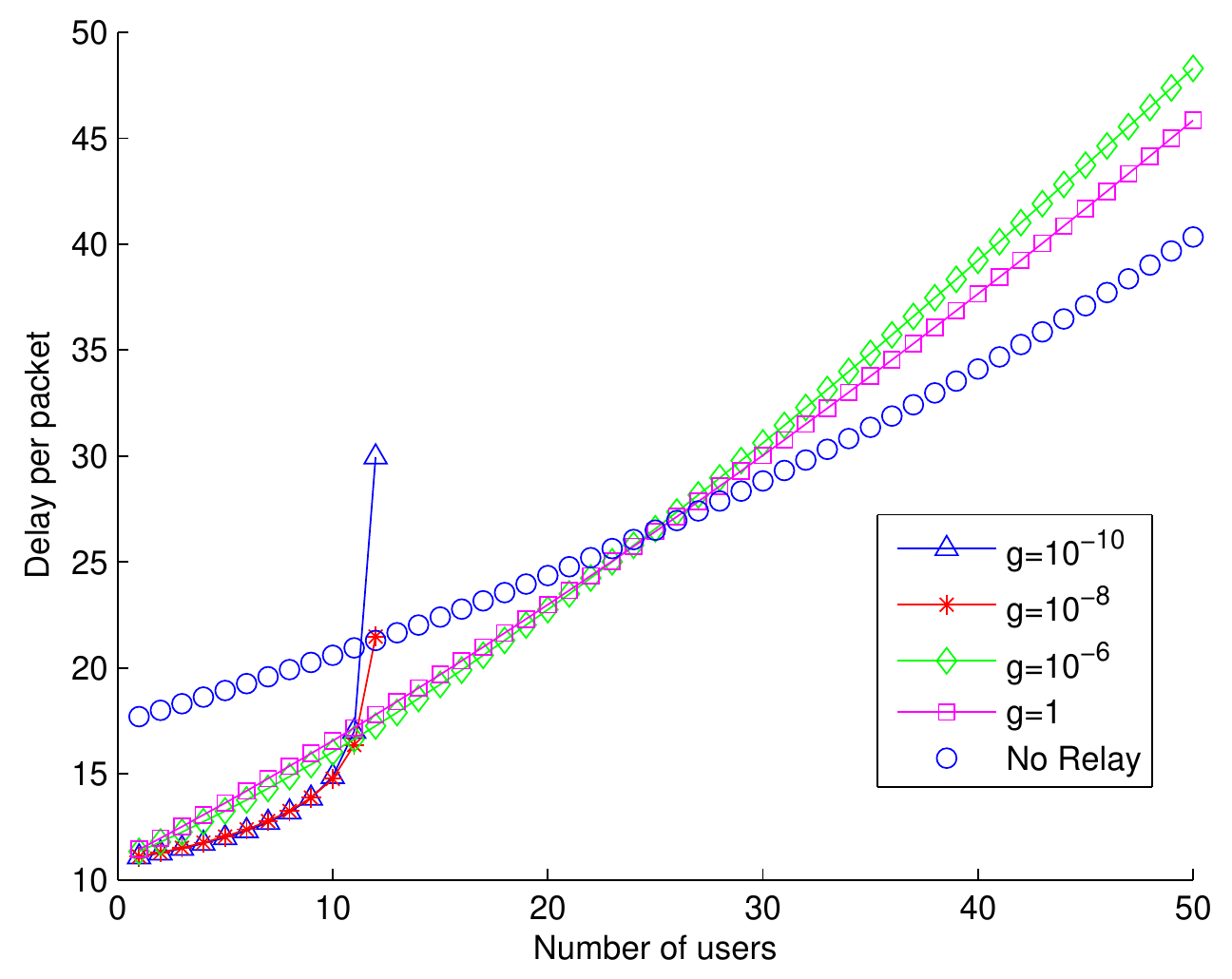}
\label{fig:Delay_n_02}
}
\caption{Average queue length and average per-packet delay vs. the number of users for $\gamma=0.2$, $q=0.1$, and $q_0 = 0.95$.}
\end{figure}

\begin{figure}[h!]
\centering
\subfigure[Average queue length vs. the number of users]{
\includegraphics[scale=0.6]{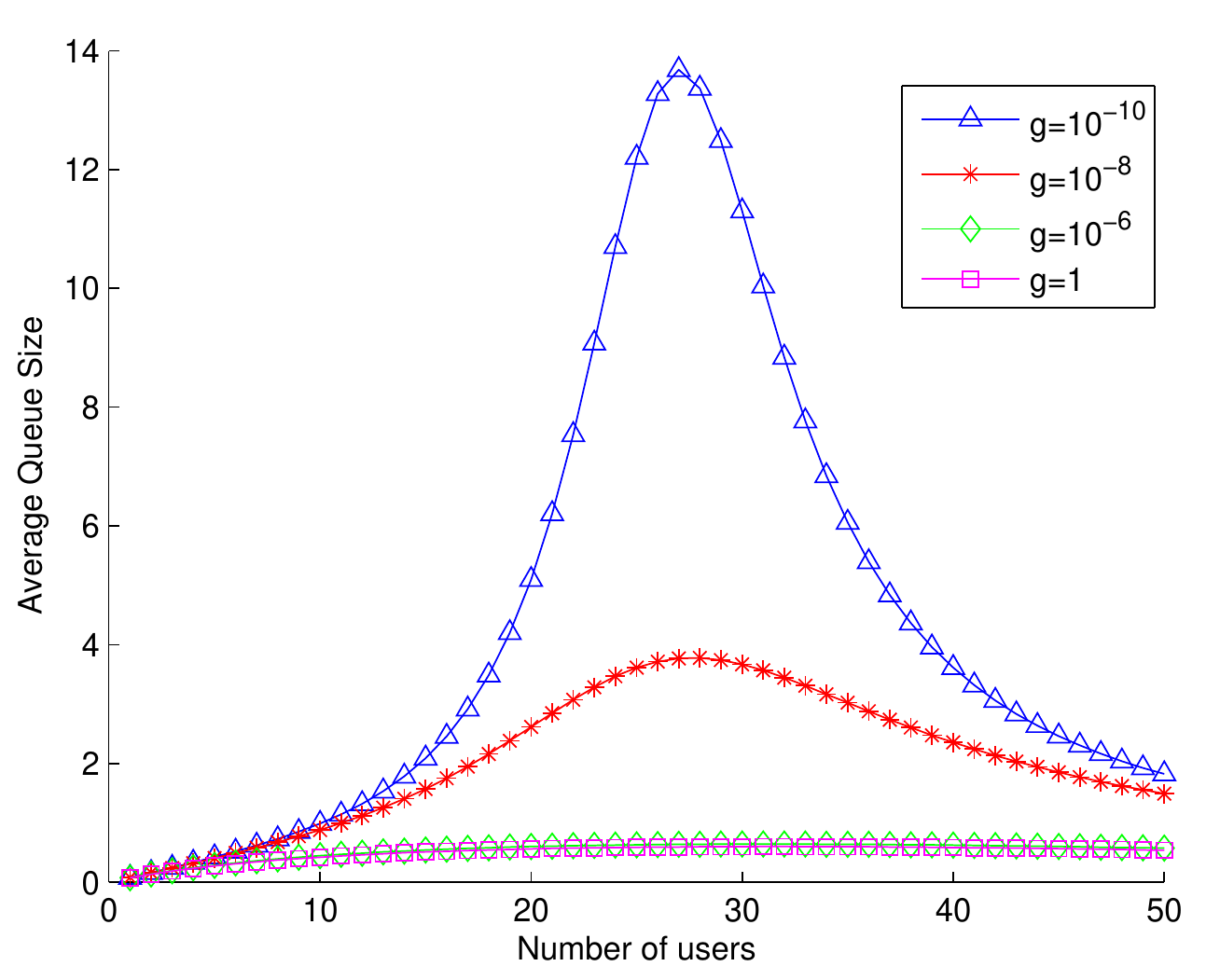}
\label{fig:avQ_n_06}
}
\subfigure[Per-packet delay vs. the number of users]{
\includegraphics[scale=0.6]{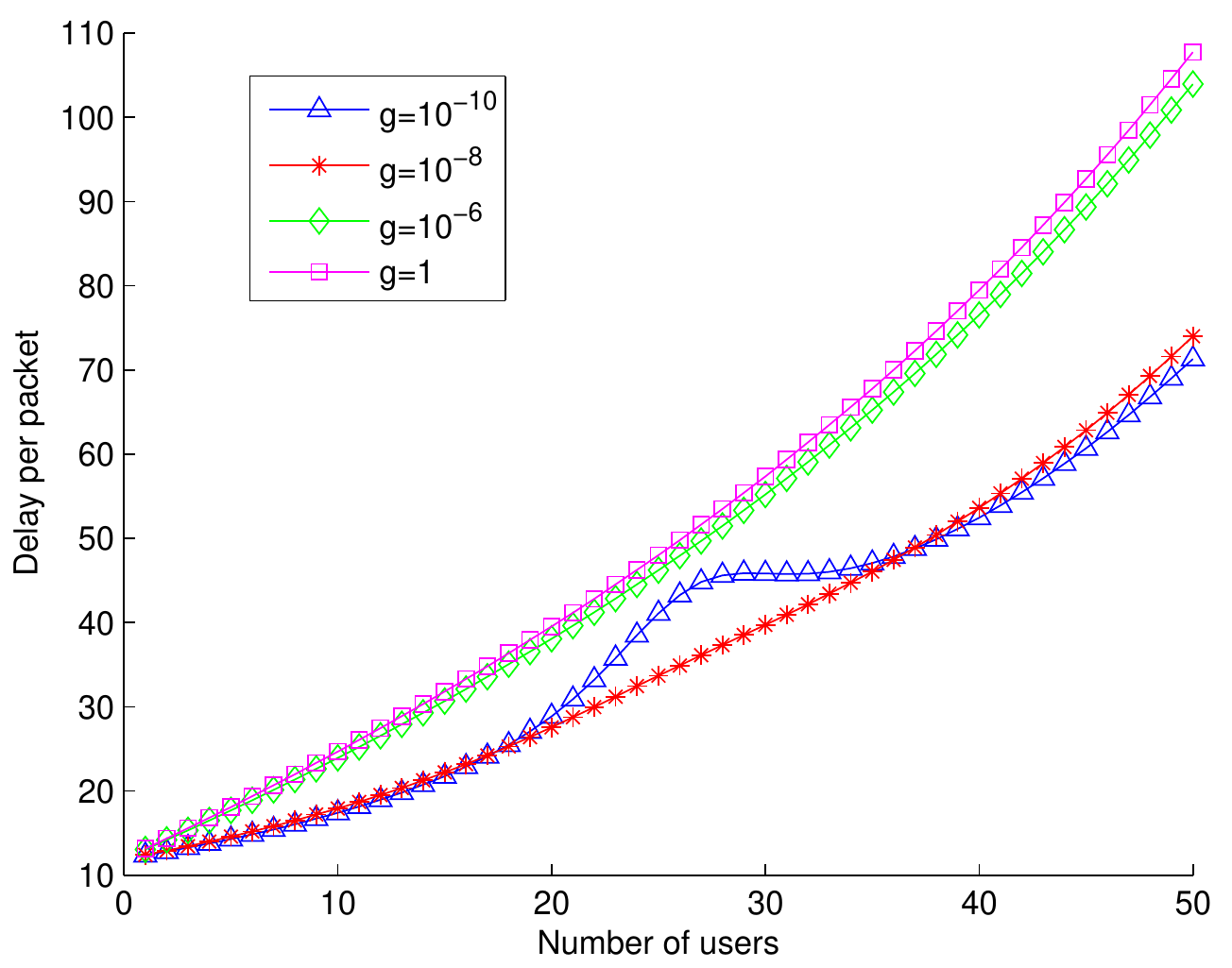}
\label{fig:Delay_n_06}
}
\caption{Average queue length and average per-packet delay vs. the number of users for $\gamma=0.6$, $q=0.1$, and $q_0 = 0.99$.}
\end{figure}

Figs.~\ref{fig:avQ_n_12} and~\ref{fig:avQ_n_25} show the average queue length for $\gamma=1.2$ and $\gamma=2.5$ respectively. In Figs.~\ref{fig:Delay_n_12} and~\ref{fig:Delay_n_25} and~\ref{fig:Delay_n_25}, we illustrate the average delay per packet.

The delay for the network without the relay is significantly large, e.g., for $\gamma=1.2$ the delay is greater than $500$ timeslots and for $\gamma=2.5$ the delay is more than $10000$ timeslots. In those cases, the existence of the relay offers significant gains not only in terms of throughput but also in the delay performance.

When we have almost perfect self-interference cancelation (except the case of $\gamma=0.2$), we observe significant gains in the delay performance compared to the case of the quasi half-duplex relay ($g \to 1$).

\begin{figure}[h!]
\centering
\subfigure[Average queue length vs. the number of users]{
\includegraphics[scale=0.6]{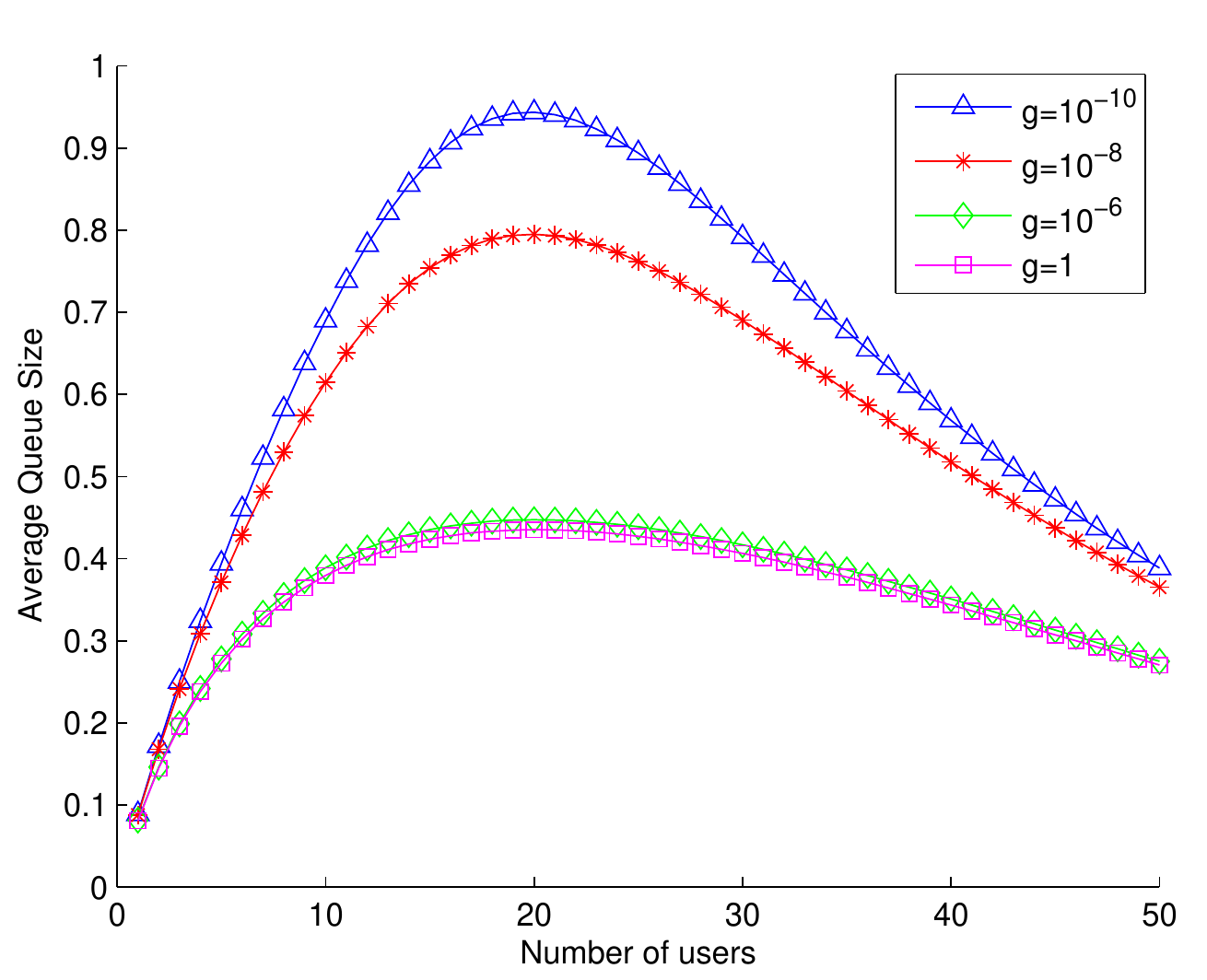}
\label{fig:avQ_n_12}
}
\subfigure[Per-packet delay vs. the number of users]{
\includegraphics[scale=0.6]{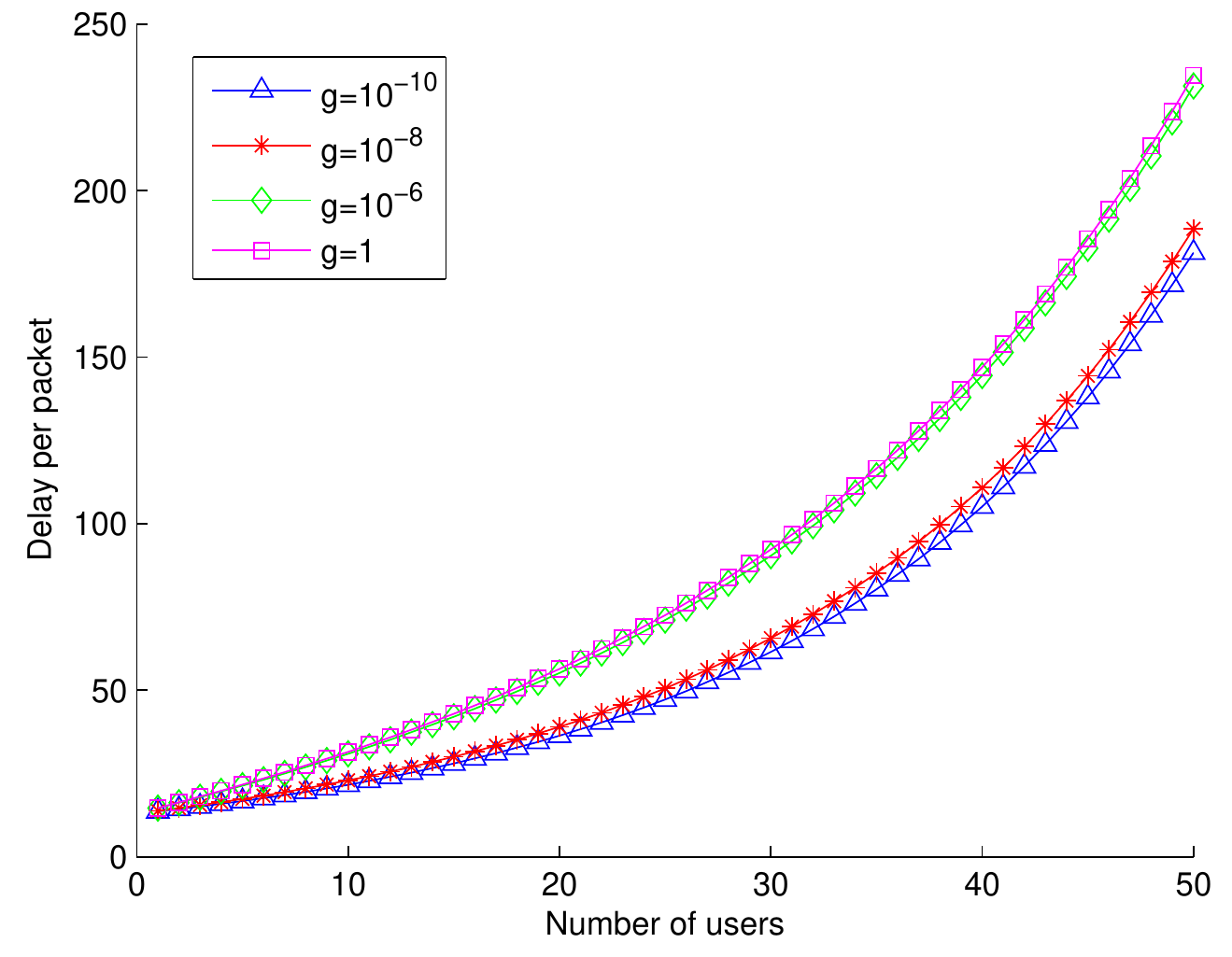}
\label{fig:Delay_n_12}
}
\caption{Average queue length and average per-packet delay vs. the number of users for $\gamma=1.2$, $q=0.1$, and $q_0 = 0.99$.}
\end{figure}

\begin{figure}[h!]
\centering
\subfigure[Average queue length vs. the number of users]{
\includegraphics[scale=0.6]{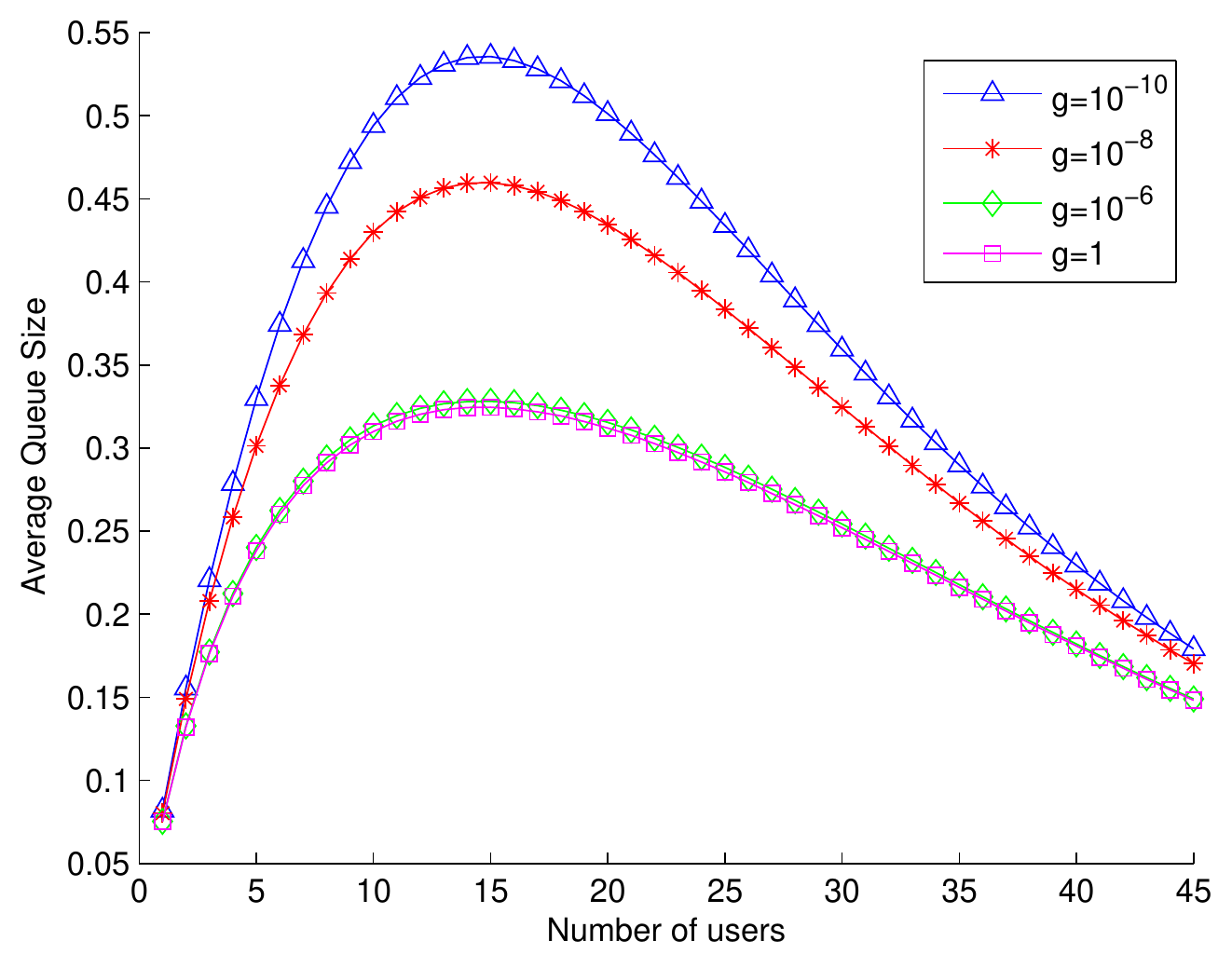}
\label{fig:avQ_n_25}
}
\subfigure[Per-packet delay vs. the number of users]{
\includegraphics[scale=0.6]{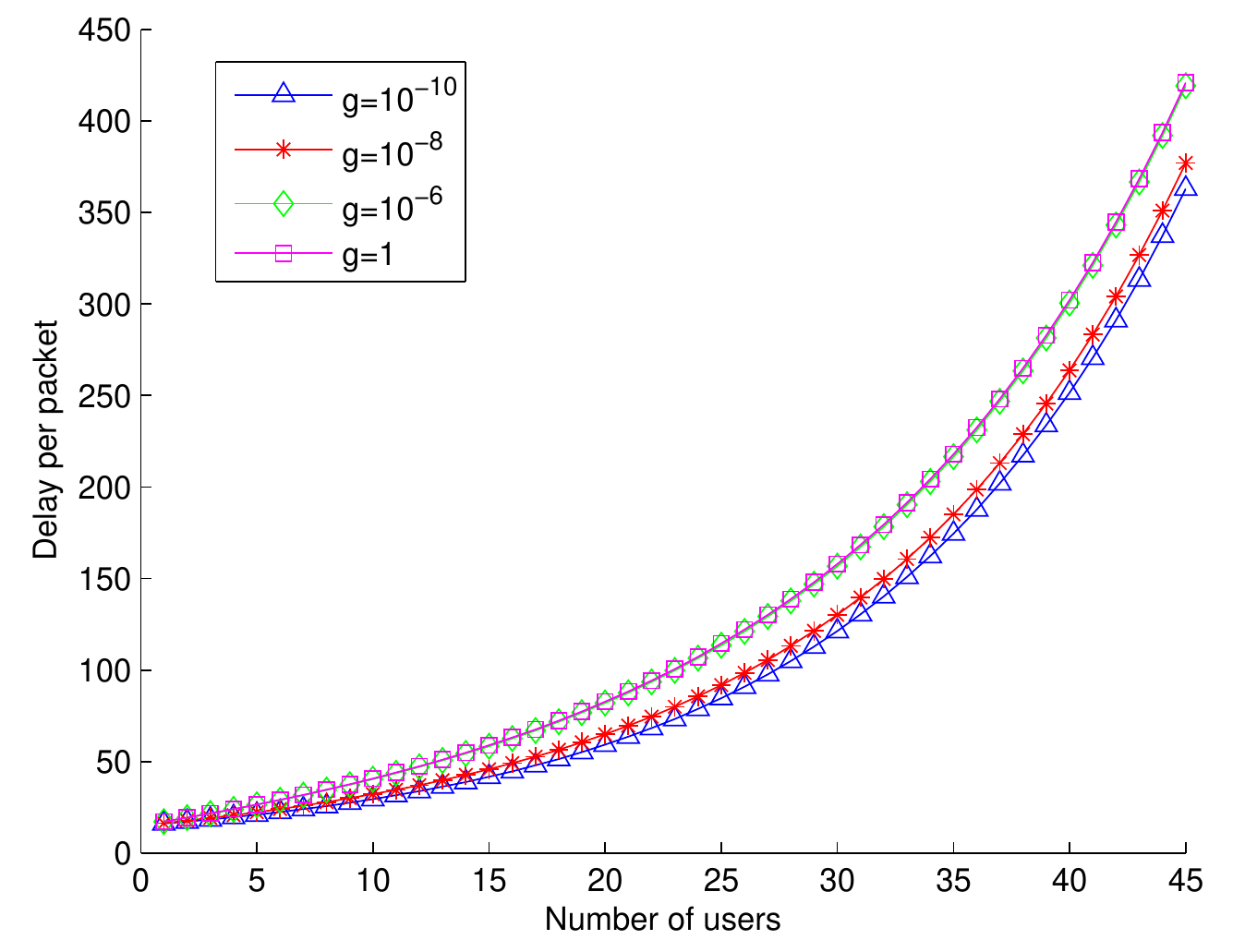}
\label{fig:Delay_n_25}
}
\caption{Average queue length and average per-packet delay vs. the number of users for $\gamma=2.5$, $q=0.1$, and $q_0 = 0.99$.}
\end{figure}

\section{Conclusions}
\label{sec:conclusions}
In this paper, we explored full-duplex communication in which a cooperative node relays packets from a number of sources to a common destination node in a random access network with multi-packet reception capability for both the relay and the destination node. Considering a multiple capture model and the self-interference due to full-duplex relay operation, a transmission is successful if the received SINR is above a certain threshold $\gamma$.

We provided analytical expressions for the performance of the relay queue, namely stability conditions, arrival and service rates, and average queue length. We studied the per-user and the aggregate throughput, and showed that the per-user throughput does not depend on the relay transmit probability under stability conditions. We also studied the impact of the self-interference coefficient $g$ on the per-user throughput, the network-wide throughput, and the average per-packet delay. We showed that the self-interference coefficient plays a crucial role when $\gamma$ is small (and $g$ tends to zero) since it may result in an unstable queue. However, for large $\gamma$ values and perfect self-interference cancelation, the gains in terms of throughput and delay are more pronounced.

Future extensions of this work may include users with non-saturated queues, i.e. sources with external random arrivals, as well as scenarios where the cooperative relay node has packets on its own and different service priorities.

\appendices

\section{Proof of Theorem~\ref{thm:2users}} \label{sec:app_proof_2}

We provide here the proof of Theorem~\ref{thm:2users}, which presents the main result for the relay queue characteristics for the
two-user case.

\underline{\textbf{\emph{Analysis of the average arrival and service rate:}}}

The average service rate $\mu$, is given by (\ref{eq:m2}), where $q_{0}$ is the transmit probability of the relay given that it has packets in its queue, and $q_{i}$ for $i \neq 0$ is the transmit probability for the $i$-th user. The term $P_{i/i,k}^{j}$ is the success probability of link $ij$ when the transmitting nodes are $i$ and $k$ and can be calculated based on (\ref{eq:succprob}).

The average arrival rate $\lambda$ of the queue is given by $\lambda=P\left(Q=0\right)\lambda_{0}+P\left(Q>0\right)\lambda_{1}$,
where $\lambda_{0}$ is the average arrival rate at the relay queue when the queue is empty and $\lambda_{1}$ when it is not.
$\lambda_{0}=r_{1}^{0}+2r_{2}^{0}$, where $r_{i}^{0}$ is the probability of receiving $i$ packets given that the queue is empty.
Accordingly, $\lambda_{1}=r_{1}^{1}+2r_{2}^{1}$, where $r_{i}^{1}$ is the probability of receiving $i$ packets when the queue is not empty.

The expressions for $r_{i}^{0}$ are given by

\begin{equation}
\label{eq:r10}
{\scriptsize
\begin{aligned}
r_{1}^{0}=q_{1}(1-q_{2})(1-P_{1/1}^{d})P_{1/1}^{0}+q_{2}(1-q_{1})(1-P_{2/2}^{d})P_{2/2}^{0}+\\
+q_{1}q_{2}(1-P_{1/1,2}^{d})P_{1/1,2}^{0}P_{2/1,2}^{d}+q_{1}q_{2}(1-P_{2/1,2}^{d})P_{2/1,2}^{0}P_{1/1,2}^{d}+\\+q_{1}q_{2}(1-P_{1/1,2}^{d})P_{1/1,2}^{0}(1-P_{2/1,2}^{d})(1-P_{2/1,2}^{0})+\\
+q_{1}q_{2}(1-P_{2/1,2}^{d})P_{2/1,2}^{0}(1-P_{1/1,2}^{d})(1-P_{1/1,2}^{0}),
\end{aligned}
}
\end{equation}
\begin{equation}
\label{eq:r20}
\begin{aligned}
r_{2}^{0}=q_{1}q_{2}(1-P_{1/1,2}^{d})(1-P_{2/1,2}^{d})P_{1/1,2}^{0}P_{2/1,2}^{0}.
\end{aligned}
\end{equation}
In order to compute for instance $r_{1}^{0}$ (i.e., the relay receives one packet), we have to take into account all the possible combinations, which are either the received packet is transmitted by the first or the second user (with all the possible combinations of active/idle users). When the relay queue is not empty, the expressions for the $r_{i}^{1}$ are given by

\begin{equation}
{\scriptsize
\label{eq:r11}
\begin{aligned}
r_{1}^{1}=(1-q_{0})q_{1}(1-q_{2})(1-P_{1/1}^{d})P_{1/1}^{0}+q_{0}q_{1}(1-q_{2})(1-P_{1/0,1}^{d})P_{1/0,1}^{0}+\\
+(1-q_{0})q_{2}(1-q_{1})(1-P_{2/2}^{d})P_{2/2}^{0}+q_{0}q_{2}(1-q_{1})(1-P_{2/0,2}^{d})P_{2/0,2}^{0}+\\
+(1-q_{0})q_{1}q_{2}(1-P_{1/1,2}^{d})P_{1/1,2}^{0}(1-P_{2/1,2}^{d})(1-P_{2/1,2}^{0})+\\
+q_{0}q_{1}q_{2}(1-P_{1/0,1,2}^{d})P_{1/0,1,2}^{0}(1-P_{2/0,1,2}^{d})(1-P_{2/0,1,2}^{0})+\\
+(1-q_{0})q_{1}q_{2}(1-P_{1/1,2}^{d})P_{1/1,2}^{0}P_{2/1,2}^{d}+\\
+q_{0}q_{1}q_{2}(1-P_{1/0,1,2}^{d})P_{1/0,1,2}^{0}P_{2/0,1,2}^{d}+\\+(1-q_{0})q_{1}q_{2}(1-P_{2/1,2}^{d})P_{2/1,2}^{0}(1-P_{1/1,2}^{d})(1-P_{1/1,2}^{0})+\\
+q_{0}q_{1}q_{2}(1-P_{2/0,1,2}^{d})P_{2/0,1,2}^{0}(1-P_{1/0,1,2}^{d})(1-P_{1/0,1,2}^{0})+\\
+(1-q_{0})q_{1}q_{2}(1-P_{2/1,2}^{d})P_{2/1,2}^{0}P_{1/1,2}^{d}+\\+q_{0}q_{1}q_{2}(1-P_{2/0,1,2}^{d})P_{2/0,1,2}^{0}P_{1/0,1,2}^{d},
\end{aligned}
}
\end{equation}

\begin{equation}
\label{eq:r21}
\begin{aligned}
r_{2}^{1}=(1-q_{0})q_{1}q_{2}(1-P_{1/1,2}^{d})P_{1/1,2}^{0}(1-P_{2/1,2}^{d})P_{2/1,2}^{0}+\\
+q_{0}q_{1}q_{2}(1-P_{1/0,1,2}^{d})P_{1/0,1,2}^{0}(1-P_{2/0,1,2}^{d})P_{2/0,1,2}^{0}.
\end{aligned}
\end{equation}

\begin{figure}[h!]
\centering
\includegraphics[scale=0.43]{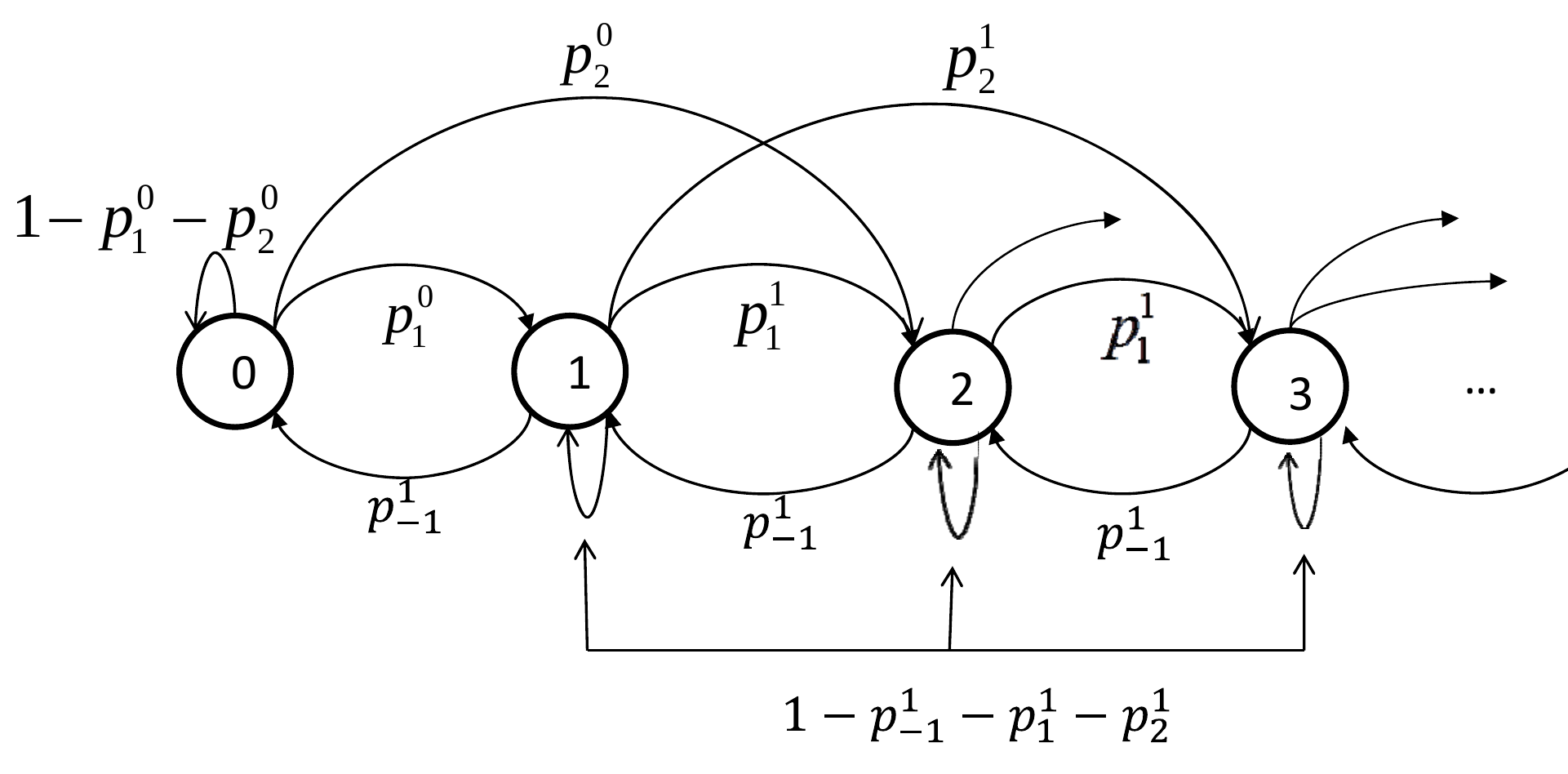}
\caption{The Markov Chain model for the two-user case.}
\label{fig:mc_rxtx}
\end{figure}

In order to fully characterize the average arrival rate at the relay, we have to compute the probability the queue is empty. We model the queue at the relay as a discrete time Markov Chain (DTMC), which describes the queue evolution and is presented in Fig.~\ref{fig:mc_rxtx}. Each state is denoted by an integer and represents the queue size at the relay node. The transition matrix of the above DTMC is a lower Hessenberg matrix given by
\begin{equation}
P=\left(\begin{array}{ccccc}
a_{0} & b_{0} & 0  & 0    & \cdots   \\
a_{1} & b_{1} & b_{0} & 0 & \cdots    \\
a_{2} & b_{2} & b_{1} & b_{0} & \cdots \\
0 & b_{3} & b_{2} & b_{1} & \cdots \\
0 & 0 & b_{3} & b_{2} & \cdots\\
\vdots & \vdots & \vdots & \vdots & \ddots
\end{array} \right) .
\end{equation}

The elements of the matrix P, are $a_{0}=1-p_{1}^{0}-p_{2}^{0},a_{1}=p_{1}^{0},a_{2}=p_{2}^{0}$, $b_{0}=p_{-1}^{1}$ and $b_{i+1}=p_{i}^{1} \text{ }  i=0,1,2,3$. The quantity $p_{i}^{0}$ ($p_{i}^{1}$) is the probability that the queue size increases by $i$ packets when the queue is empty (not empty). Note that $p_{i}^{0}=r_{i}^{0}$ because when the queue is empty, the probability of $i$ packets arriving is the same with the probability that the queue size increases by $i$ packets; when the queue is not empty however, this is not true. For example the probability of two packets arriving is not the same with the probability of doubling the queue size; this is because both arrivals and departures can occur at the same time. The expressions for the $p_{i}^{j}$ are given by

\begin{equation}
{\scriptsize
\label{eq:p-11}
\begin{aligned}
p_{-1}^{1}=q_{0}(1-q_{1})(1-q_{2})P_{0/0}^{d}+q_{0}(1-q_{1})q_{2}P_{0/0,2}^{d}P_{2/0,2}^{d}+\\
+q_{0}(1-q_{1})q_{2}P_{0/0,2}^{d}(1-P_{2/0,2}^{d})(1-P_{2/0,2}^{0})+q_{0}q_{1}(1-q_{2})P_{0/0,1}^{d}P_{1/0,1}^{d}+\\
+q_{0}q_{1}(1-q_{2})P_{0/0,1}^{d}(1-P_{1/0,1}^{d})(1-P_{1/0,1}^{0})+\\+q_{0}q_{1}q_{2}P_{0/0,1,2}^{d}P_{1/0,1,2}^{d}P_{2/0,1,2}^{d}+\\
+q_{0}q_{1}q_{2}P_{0/0,1,2}^{d}(1-P_{1/0,1,2}^{d})(1-P_{1/0,1,2}^{0})(1-P_{2/0,1,2}^{d})(1-P_{2/0,1,2}^{0})+\\
+q_{0}q_{1}q_{2}P_{0/0,1,2}^{d}P_{1/0,1,2}^{d}(1-P_{2/0,1,2}^{d})(1-P_{2/0,1,2}^{0})+\\
+q_{0}q_{1}q_{2}P_{0/0,1,2}^{d}(1-P_{1/0,1,2}^{d})(1-P_{1/0,1,2}^{0})P_{2/0,1,2}^{d},
\end{aligned}
}
\end{equation}
\begin{equation}
\label{eq:p01}
\begin{aligned}
p_{0}^{1}=1-p_{-1}^{1}-p_{1}^{1}-p_{2}^{1},
\end{aligned}
\end{equation}

\begin{equation}
{\scriptsize
\label{eq:p11}
\begin{aligned}
p_{1}^{1}=(1-q_{0})q_{1}(1-q_{2})(1-P_{1/1}^{d})P_{1/1}^{0}+(1-q_{0})q_{1}q_{2}(1-P_{1/1,2}^{d})P_{1/1,2}^{0}P_{2/1,2}^{d}+\\
+(1-q_{0})q_{1}q_{2}(1-P_{1/1,2}^{d})P_{1/1,2}^{0}(1-P_{2/1,2}^{d})(1-P_{2/1,2}^{0})+\\
+(1-q_{0})(1-q_{1})q_{2}(1-P_{2/2}^{d})P_{2/2}^{0}+(1-q_{0})q_{1}q_{2}(1-P_{2/1,2}^{d})P_{2/1,2}^{0}P_{1/1,2}^{d}+\\
+(1-q_{0})q_{1}q_{2}(1-P_{2/1,2}^{d})P_{2/1,2}^{0}(1-P_{1/1,2}^{d})(1-P_{1/1,2}^{0})+\\
+q_{0}q_{1}q_{2}P_{0/0,1,2}^{d}(1-P_{1/0,1,2}^{d})P_{1/0,1,2}^{0}(1-P_{2/0,1,2}^{d})P_{2/0,1,2}^{0}+\\
+q_{0}q_{1}(1-q_{2})(1-P_{0/0,1}^{d})(1-P_{1/0,1}^{d})P_{1/0,1}^{0}+\\
+q_{0}q_{1}q_{2}(1-P_{0/0,1,2}^{d})(1-P_{1/0,1,2}^{d})P_{1/0,1,2}^{0}P_{2/0,1,2}^{d}+\\
+q_{0}q_{1}q_{2}(1-P_{0/0,1,2}^{d})(1-P_{1/0,1,2}^{d})P_{1/0,1,2}^{0}(1-P_{2/0,1,2}^{d})(1-P_{2/0,1,2}^{0})+\\
+q_{0}q_{2}(1-q_{1})(1-P_{0/0,2}^{d})(1-P_{2/0,2}^{d})P_{2/0,2}^{0}+\\
+q_{0}q_{1}q_{2}(1-P_{0/0,1,2}^{d})(1-P_{2/0,1,2}^{d})P_{2/0,1,2}^{0}P_{1/0,1,2}^{d}+\\
+q_{0}q_{1}q_{2}(1-P_{0/0,1,2}^{d})(1-P_{2/0,1,2}^{d})P_{2/0,1,2}^{0}(1-P_{1/0,1,2}^{d})(1-P_{1/0,1,2}^{0}),
\end{aligned}
}
\end{equation}
\begin{equation}
\label{eq:p21}
\begin{aligned}
p_{2}^{1}=(1-q_{0})q_{1}q_{2}(1-P_{1/1,2}^{d})P_{1/1,2}^{0}(1-P_{2/1,2}^{d})P_{2/1,2}^{0}+\\
+q_{0}q_{1}q_{2}(1-P_{0/0,1,2}^{d})(1-P_{1/0,1,2}^{d})P_{1/0,1,2}^{0}(1-P_{2/0,1,2}^{d})P_{2/0,1,2}^{0}.
\end{aligned}
\end{equation}
Note that $p_{i}^{j}$ and $r_{i}^{j}$ are in general different quantities, however in half-duplex relay systems we have $p_{i}^{j}=r_{i}^{j}$.

The difference equations that govern the evolution of the states are given by
\begin{equation}
Ps=s \Rightarrow s_{i}=a_{i}s_{0}+\sum_{j=1}^{i+1}{b_{i-j+1}s_{j}}.
\end{equation}
We apply the z-transform technique to compute the steady state distribution, i.e., we let
\begin{equation}
A(z)=\sum_{i=0}^{2}{a_{i}z^{-i}}, B(z)=\sum_{i=0}^{3}{b_{i}z^{-i}}, S(z)=\sum_{i=0}^{\infty}{s_{i}z^{-i}}.
\end{equation}
We know that~\cite{b:Gebali}
\begin{equation}
S(z)=s_{0}\frac{z^{-1}A(z)-B(z)}{z^{-1}-B(z)}.
\end{equation}
It is also known that the probability of the queue in the relay is empty is given by~\cite{b:Gebali}
\begin{equation}
\label{eq:P(Q=0)}
P\left(Q=0\right)=\frac{1+B^{'}(1)}{1+B^{'}(1)-A^{'}(1)}.
\end{equation}
The expressions of $A^{'}(1)$ and $B^{'}(1)$ are:
\begin{equation}
\label{eq:A'(1)}
\begin{aligned}
A^{'}(z)=\left(\sum_{i=0}^{2}{a_{i}z^{-i}} \right)^{'}=-\sum_{i=1}^{2}{ia_{i}z^{-(i+1)}} \\
\Rightarrow A^{'}(1)=-\sum_{i=1}^{2}{ia_{i}}\Rightarrow A^{'}(1)=-\sum_{i=1}^{2}{ip_{i}^{0}}=-\lambda_{0},
\end{aligned}
\end{equation}
\begin{equation}
\label{eq:B'(1)}
\begin{aligned}
B^{'}(z)=\left(\sum_{i=0}^{3}{b_{i}z^{-i}} \right)^{'}=-\sum_{i=0}^{3}{ib_{i}z^{-(i+1)}} \\ \Rightarrow B^{'}(1)=-\sum_{i=0}^{3}{ib_{i}}=-1+p_{-1}^{1}-p_{1}^{1}-2p_{2}^{1}.
\end{aligned}
\end{equation}
Then, the probability of the queue in the relay is empty is given by (\ref{eq:probempty2}).
Therefore, the average arrival rate $\lambda$ is given by (\ref{eq:lambda2}).

\underline{\textbf{\emph{Average Queue Length:}}}
The average queue length is known to be $\overline{Q}=-S^{'}(1)$, where $S^{'}(1)=s_{0}\frac{K^{''}(1)}{L^{''}(1)}$~\cite{b:Gebali}.
The expressions for $K(z)$ and $L(z)$ are given by
\begin{equation}
\label{eq:K(z)}
\begin{aligned}
K(z)=\left(-z^{-2}A(z)+z^{-1}A^{'}(z)-B^{'}(z) \right) \left(z^{-1}-B(z)\right) - \\ - \left(z^{-1}A(z)-B(z) \right) \left(-z^{-2}-B^{'}(z) \right),
\end{aligned}
\end{equation}
\begin{equation}
\label{eq:L(z)}
\begin{aligned}
L(z)=\left(z^{-1}-B(z) \right)^{2}.
\end{aligned}
\end{equation}
Then $K^{''}(1)$ and $L^{''}(1)$ are given by
\begin{equation}
\label{eq:K''(1)}
{\scriptsize
\begin{aligned}
K^{''}(1)=\left(2A(1)-2A^{'}(1)+A^{''}(1)-B^{''}(1) \right) \left(-1-B^{'}(1) \right)-\\ - \left(2-B^{''}(1) \right) \left(-A(1)+A^{'}(1)-B^{'}(1) \right),
\end{aligned}
}
\end{equation}
\begin{equation}
\label{eq:L''(1)}
\begin{aligned}
L^{''}(z)=\left[2\left(z^{-1}-B(z) \right) \left(-z^{-2}-B^{'}(z) \right) \right]^{'} \\ \Rightarrow L^{''}(1)=2\left(-1-B^{'}(1)\right)^{2}.
\end{aligned}
\end{equation}
The values of $A^{''}(1)$ and $B^{''}(1)$ are:
\begin{equation}
\label{eq:A''(1)}
\begin{aligned}
A^{''}(z)=\left(-\sum_{i=1}^{2}{ia_{i}z^{-(i+1)}} \right)^{'}=\sum_{i=1}^{2}{i(i+1)a_{i}z^{-(i+2)}} \\ \Rightarrow A^{''}(1)=2p_{1}^{0}+6p_{2}^{0},
\end{aligned}
\end{equation}
\begin{equation}
\label{eq:B''(1)}
\begin{aligned}
B^{''}(z)=\left(-\sum_{i=1}^{3}{ib_{i}z^{-(i+1)}} \right)^{'}=\sum_{i=1}^{3}{i(i+1)b_{i}z^{-(i+2)}} \\ \Rightarrow B^{''}(1)=2-2p_{-1}^{1}+4p_{1}^{1}+10p_{2}^{1}.
\end{aligned}
\end{equation}
The average queue length is given by (\ref{eq:avQ2}).

\underline{\textbf{\emph{Condition for the stability of the queue:}}}
An important tool to determine stability is Loyne's criterion~\cite{b:Loynes}, which states that if the arrival and service processes of a queue are jointly strictly stationary and ergodic, the queue is stable if and only if the average arrival rate is strictly less than the average service rate. If the queue is stable, the departure rate (throughput) is equal to the arrival rate, i.e.,
$\lambda_{1}<\mu\Leftrightarrow r_{1}^{1}+2r_{2}^{1}<\mu$ where $r_{1}^{1}=(1-q_{0})A_{1}+q_{0}B_{1}$, $r_{2}^{1}=(1-q_{0})A_{2}+q_{0}B_{2}$ and $\mu=q_{0}A$.

The expressions for $A,A_{i},B_{i}$ are given by

\begin{equation*}
\label{eq:A1B1}
\begin{aligned}
A_{1}=q_{1}(1-q_{2})(1-P_{1/1}^{d})P_{1/1}^{0}+q_{2}(1-q_{1})(1-P_{2/2}^{d})P_{2/2}^{0}+\\
+q_{1}q_{2}(1-P_{1/1,2}^{d})P_{1/1,2}^{0}(1-P_{2/1,2}^{d})(1-P_{2/1,2}^{0})+\\+q_{1}q_{2}(1-P_{1/1,2}^{d})P_{1/1,2}^{0}P_{2/1,2}^{d}+\\
+q_{1}q_{2}(1-P_{2/1,2}^{d})P_{2/1,2}^{0}(1-P_{1/1,2}^{d})(1-P_{1/1,2}^{0})+\\+q_{1}q_{2}(1-P_{2/1,2}^{d})P_{2/1,2}^{0}P_{1/1,2}^{d},
\end{aligned}
\end{equation*}
\begin{equation*}
{\scriptsize
\begin{aligned}
B_{1}=q_{1}(1-q_{2})(1-P_{1/0,1}^{d})P_{1/0,1}^{0}+q_{2}(1-q_{1})(1-P_{2/0,2}^{d})P_{2/0,2}^{0}+\\
+q_{1}q_{2}(1-P_{1/0,1,2}^{d})P_{1/0,1,2}^{0}(1-P_{2/0,1,2}^{d})(1-P_{2/0,1,2}^{0})+\\+q_{1}q_{2}(1-P_{1/0,1,2}^{d})P_{1/0,1,2}^{0}P_{2/0,1,2}^{d}+\\
+q_{1}q_{2}(1-P_{2/0,1,2}^{d})P_{2/0,1,2}^{0}(1-P_{1/0,1,2}^{d})(1-P_{1/0,1,2}^{0})+\\+q_{1}q_{2}(1-P_{2/0,1,2}^{d})P_{2/0,1,2}^{0}P_{1/0,1,2}^{d},
\end{aligned}
}
\end{equation*}
\begin{equation*}
\label{eq:A2B2}
\begin{aligned}
A_{2}=q_{1}q_{2}(1-P_{1/1,2}^{d})P_{1/1,2}^{0}(1-P_{2/1,2}^{d})P_{2/1,2}^{0},\\ B_{2}=q_{1}q_{2}(1-P_{1/0,1,2}^{d})P_{1/0,1,2}^{0}(1-P_{2/0,1,2}^{d})P_{2/0,1,2}^{0},
\end{aligned}
\end{equation*}
\begin{equation*}
\label{eq:A}
\begin{aligned}
A=(1-q_{1})(1-q_{2})P_{0/0}^{d}+q_{1}(1-q_{2})P_{0/0,1}^{d}+\\+q_{2}(1-q_{1})P_{0/0,2}^{d}+q_{1}q_{2}P_{0/0,1,2}^{d}.
\end{aligned}
\end{equation*}
Then the values of $q_{0}$ for which the queue is stable are given by $q_{0min}<q_{0}<1$, where
\begin{equation}
\label{eq:q0min2}
q_{0min}=\frac{A_{1}+2A_{2}}{A+A_{1}+2A_{2}-B_{1}-2B_{2}}.
\end{equation}

\section{Proof of Theorem~\ref{thm:nusers}} \label{sec:app_proof_n}

In this appendix, we provide the proof of Theorem~\ref{thm:nusers}, which presents the relay's queue characteristics for the symmetric $n$-user case.

\underline{\textbf{\emph{Computation of the average arrival and service rate:}}}
The service rate is given by (\ref{eq:mun}).
The average arrival rate $\lambda$ of the queue is $\lambda=P\left(Q=0\right)\lambda_{0}+P\left(Q>0\right)\lambda_{1}$,
with $\lambda_{0}=\sum_{k=1}^{n}{kr_{k}^{0}}$, where the $r_{k}^{0}$ is the probability that the relay received $k$ packets when the queue is empty, given by
\begin{equation}
\label{eq:rk0n}
\begin{aligned}
r_{k}^{0}=\sum_{i=k}^{n}{{n \choose i}{i \choose k} {q^{i}(1-q)^{n-i}}P_{0,i,0}^{k}} \times \\ \times  \left(1-P_{d,i,0}\right)^{k}\left[1-P_{0,i,0}(1-P_{d,i,0})\right]^{i-k},\text{ }1 \leq k \leq n.
\end{aligned}
\end{equation}
$\lambda_{1}=\sum_{k=1}^{n}{kr_{k}^{1}}$, where the $r_{k}^{1}$ is the probability that the relay received $k$ packets when the queue is not empty and is given by
\begin{equation}
\label{eq:rk1n}
\begin{aligned}
r_{k}^{1}=(1-q_{0})\sum_{i=k}^{n}{{n \choose i}{i \choose k} {q^{i}(1-q)^{n-i}}P_{0,i,0}^{k}} \times \\ \times \left(1-P_{d,i,0}\right)^{k}\left[1-P_{0,i,0}(1-P_{d,i,0})\right]^{i-k}+\\
+q_{0}\sum_{i=k}^{n}{{n \choose i}{i \choose k} {q^{i}(1-q)^{n-i}}P_{0,i,1}^{k}\left(1-P_{d,i,1}\right)^{k}} \times \\ \times \left[1-P_{0,i,1}(1-P_{d,i,1})\right]^{i-k},\text{ }1 \leq k \leq n.
\end{aligned}
\end{equation}
The elements of the transition matrix are given by
$a_{k}=p_{k}^{0}$, $b_{0}=p_{-1}^{1}$, $b_{1}=p_{0}^{1}$ and $b_{k+1}=p_{k}^{1} \text{ } \forall k>0$ where

\begin{equation}
\label{eq:pk0n}
\begin{aligned}
p_{k}^{0}=\sum_{i=k}^{n}{{n \choose i}{i \choose k} {q^{i}(1-q)^{n-i}}P_{0,i,0}^{k}\left(1-P_{d,i,0}\right)^{k}} \times \\
\times \left[1-P_{0,i,0}(1-P_{d,i,0})\right]^{i-k},\text{ }1 \leq k \leq n,
\end{aligned}
\end{equation}
\begin{equation}
\label{eq:p-11n}
p_{-1}^{1}=q_{0}\sum_{k=0}^{n}{{n \choose k}q^{k}(1-q)^{n-k}P_{0d,k}\left[1-P_{0,k,1}(1-P_{d,k,1})\right]^{k}},
\end{equation}
\begin{equation}
\label{eq:pk1n}
\begin{aligned}
p_{k}^{1}=(1-q_{0})\sum_{i=k}^{n}{{n \choose i}{i \choose k} {q^{i}(1-q)^{n-i}}P_{0,i,0}^{k}} \times \\
\times  \left(1-P_{d,i,0}\right)^{k}\left[1-P_{0,i,0}(1-P_{d,i,0})\right]^{i-k}+\\
+q_{0}\sum_{i=k}^{n}{{n \choose i}{i \choose k} {q^{i}(1-q)^{n-i}}(1-P_{0d,i})P_{0,i,1}^{k}} \times \\
\times \left(1-P_{d,i,1}\right)^{k}\left[1-P_{0,i,1}(1-P_{d,i,1})\right]^{i-k}+\\
+q_{0}\sum_{i=k+1}^{n}{{n \choose i}{i \choose k+1} {q^{i}(1-q)^{n-i}}P_{0d,i}P_{0,i,1}^{k+1}} \times \\
\times \left(1-P_{d,i,1}\right)^{k+1}\left[1-P_{0,i,1}(1-P_{d,i,1})\right]^{i-k-1},
\end{aligned}
\end{equation}
\begin{equation}
\label{eq:p01n}
p_{0}^{1}=1-p_{-1}^{1}-\sum_{i=1}^{n}{p_{i}^{1}}.
\end{equation}
The probability that the queue in the relay is empty is given by (\ref{eq:P(Q=0)}), where the expressions for $A^{'}(1)$ and $B^{'}(1)$ are
\begin{equation}
\label{eq:A'(1)}
\begin{aligned}
A^{'}(z)=\left(\sum_{i=0}^{n}{a_{i}z^{-i}} \right)^{'}=-\sum_{i=1}^{n}{ia_{i}z^{-(i+1)}} \\
\Rightarrow A^{'}(1)=-\sum_{i=1}^{n}{ia_{i}}\Rightarrow A^{'}(1)=-\sum_{i=1}^{n}{ip_{i}^{0}}=-\lambda_{0},
\end{aligned}
\end{equation}
\begin{equation}
\label{eq:B'(1)}
\begin{aligned}
B^{'}(z)=\left(\sum_{i=0}^{n+1}{b_{i}z^{-i}} \right)^{'}=-\sum_{i=i}^{n+1}{ib_{i}z^{-(i+1)}} \\ \Rightarrow B^{'}(1)=-\sum_{i=i}^{n+1}{ib_{i}}=-b_{1}-\sum_{i=2}^{n+1}{ib_{i}}=-1+p_{-1}^{1}-\sum_{i=1}^{n}{ip_{i}^{1}}.
\end{aligned}
\end{equation}
Then the probability that the queue in the relay is empty is given by (\ref{eq:probemptyn}).

\underline{\textbf{\emph{Average Queue Length:}}}
As we showed in Appendix~\ref{sec:app_proof_2}, the average queue length is given by $\overline{Q}=-S^{'}(1)$, where $S^{'}(1)=s_{0}\frac{K^{''}(1)}{L^{''}(1)}$.
The expressions for $K^{''}(1)$ and $L^{''}(1)$ are given by (\ref{eq:K''(1)}) and (\ref{eq:L''(1)}).
The expressions for $A^{''}(1)$ and $B^{''}(1)$ are
\begin{equation}
\label{eq:A''(1)}
\begin{aligned}
A^{''}(z)=\left(-\sum_{i=1}^{n}{ia_{i}z^{-(i+1)}} \right)^{'}=\sum_{i=1}^{n}{i(i+1)a_{i}z^{-(i+2)}} \\ \Rightarrow A^{''}(1)=\sum_{i=1}^{n}{i(i+1)a_{i}}=\sum_{i=1}^{n}{i(i+1)p_{i}^{0}},
\end{aligned}
\end{equation}
\begin{equation}
\label{eq:B''(1)}
\begin{aligned}
B^{''}(z)=\left(-\sum_{i=i}^{n+1}{ib_{i}z^{-(i+1)}} \right)^{'}=\sum_{i=1}^{n+1}{i(i+1)b_{i}z^{-(i+2)}} \\ \Rightarrow B^{''}(1)=\sum_{i=1}^{n+1}{i(i+1)b_{i}}=2-2p_{-1}^{1}+\sum_{i=1}^{n}{i(i+3)p_{i}^{1}}.
\end{aligned}
\end{equation}
Following the same methodology as in Appendix~\ref{sec:app_proof_2}, we obtain that the average queue length given by (\ref{eq:avQn}).

\underline{\textbf{\emph{Condition for the stability of the queue:}}}
As in Appendix~\ref{sec:app_proof_2}, the queue is stable if $\lambda_{1}<\mu\Leftrightarrow \sum_{k=1}^{n}{kr_{k}^{1}}<\mu$, where $r_{k}^{1}=(1-q_{0})A_{k}+q_{0}B_{k}$ and $\mu=q_{0}A$.
The expressions for $A,A_{k},B_{k}$ are :

{\scriptsize
\begin{equation}
\label{eq:Akn}
\begin{aligned}
A_{k}=\sum_{i=k}^{n}{{n \choose i}{i \choose k} {q^{i}(1-q)^{n-i}}P_{0,i,0}^{k}\left(1-P_{d,i,0}\right)^{k}\left[1-P_{0,i,0}(1-P_{d,i,0})\right]^{i-k}},
\end{aligned}
\end{equation}
\begin{equation}
\label{eq:Bkn}
\begin{aligned}
B_{k}=\sum_{i=k}^{n}{{n \choose i}{i \choose k} {q^{i}(1-q)^{n-i}}P_{0,i,1}^{k}\left(1-P_{d,i,1}\right)^{k}\left[1-P_{0,i,1}(1-P_{d,i,1})\right]^{i-k}},
\end{aligned}
\end{equation}
}
\begin{equation}
\label{eq:An}
\begin{aligned}
A=\sum_{k=0}^{n}{{n \choose k} {q^{k}(1-q)^{n-k}}P_{0d,k}}.
\end{aligned}
\end{equation}
The values of $q_{0}$ for which the queue is stable are given by $q_{0min}<q_{0}<1$, where
\begin{equation}
\label{eq:q0minn}
q_{0min}=\frac{\displaystyle \sum_{k=1}^{n}{kA_{k}}}{\displaystyle A+\sum_{k=1}^{n}{kA_{k}}-\sum_{k=1}^{n}{kB_{k}}}.
\end{equation}

\section{Proof of Theorem~\ref{thm:delay}}\label{sec:app_delay}
We first present the analysis for the average delay $D_i$ required to deliver a packet from source $i$ to the destination.
This delay is the summation of the transmission delay from the source (to either the destination directly or the relay node), the queueing delay at the relay node, and the transmission delay from the relay to the destination.

When a packet is transmitted from the $i$-th source, there is a probability that this packet reaches the destination directly, which is $T_{D,i}$.
In the case that the transmission to the destination is not successful but is successful to the relay, the packet enters the relay queue, this is with probability $T_{R,i}$. The total time that the packet entering the relay queue reaches the destination is denoted by $D_Q$.
If the transmission from the source to the destination is unsuccessful to both destination and relay nodes, then it remains at
the source for future retransmission (with probability $1-T_{D,i}-T_{R,i}$).

The average delay $D_i$ is given by $D_i=T_{D,i}+T_{R,i}\left(1+D_{R}\right)+(1-T_{D,i}-T_{R,i})\left(1+D_i \right)$, which after some simplifications results in
\begin{equation} \label{eq:Di}
D_i=\frac{1+T_{R,i}D_{R}}{T_i}.
\end{equation}
The expressions for $T_{R,i}$, $T_{D,i}$, and $T_i$, are given in Section~\ref{sec:thr_analysis}.

When the packet from source $i$ that enters the queue waits, while other packets in the queue are transmitted, this waiting time
is the queue delay and is denoted by $D_Q$. When the packet that waits at the queue reaches the head of the queue, then it is transmitted from the relay with a probability
(due to the random access assumption), the transmission delay from the relay to the destination is $\frac{1}{\mu}$, where $\mu$ is the service rate.
The total delay in the relay node is denoted by $D_{R}$. The expression for $D_R$ is $D_{R}=D_{Q}+\mu+\left(1-\mu \right) \left(1+\frac{1}{\mu} \right)$, which is
\begin{align} \label{eq:DR}
D_{R}=D_{Q}+\frac{1}{\mu}.
\end{align}
From Little's law, we obtain that $D_{Q}=\frac{\overline{Q}}{\lambda}$, where $\overline{Q}$ is the average queue length for the relay and $\lambda$ is the average arrival rate. The expressions for $\overline{Q}$ and $\lambda$ are presented in Section~\ref{sec:analysis}.
After substituting (\ref{eq:DR}) into (\ref{eq:Di}), we obtain (\ref{eq:delay}) in Theorem~\ref{thm:delay}.
Note that in our study we do not take into account the processing and the propagation delay.

\bibliographystyle{ieeetr}
\bibliography{thesis}

\end{document}